\documentclass[aps,prl,twocolumn,superscriptaddress,nofootinbib,floatfix]{revtex4-2}

\usepackage[T1]{fontenc}
\usepackage{amsmath,amssymb,mathtools}
\usepackage{amsthm}
\usepackage{mathrsfs}
\usepackage{graphicx}
\usepackage{algorithm}
\usepackage{algpseudocode}
\usepackage{xcolor}
\usepackage[colorlinks=true,allcolors=blue!60!black]{hyperref}

\makeatletter
\let\smtoc@section\l@section
\let\smtoc@subsection\l@subsection
\newcommand{\smtocstart}{%
  \let\l@section\smtoc@section
  \let\l@subsection\smtoc@subsection
}
\newcommand{\supplementtableofcontents}{%
  \begingroup
  \let\l@section\@gobble@tw@
  \let\l@subsection\@gobble@tw@
  \let\l@subsubsection\@gobble@tw@
  \tableofcontents
  \endgroup
}
\makeatother

\newcommand{\HH}{\mathcal{H}}
\newcommand{\Rbb}{\mathbb{R}}
\newcommand{\Cbb}{\mathbb{C}}
\newcommand{\OO}{\mathcal{O}}
\newcommand{\HL}{\mathcal{H}_{L}}
\newcommand{\HC}{\mathcal{H}_{C}}
\newcommand{\HS}{\mathcal{H}_{S}}
\newcommand{\HE}{\mathcal{H}_{E}}
\newcommand{\HRef}{\mathcal{H}_{L'}}
\newcommand{\Hr}{\mathcal{H}_{\rho}}
\newcommand{\NN}{\mathcal{N}}
\newcommand{\ACal}{\mathcal{A}}
\newcommand{\BCal}{\mathcal{B}}
\newcommand{\CCal}{\mathcal{C}}
\newcommand{\DD}{\mathcal{D}}
\newcommand{\DDtilde}{\widetilde{\mathcal{D}}}
\newcommand{\Dtilde}{\widetilde{D}}
\newcommand{\KK}{\mathcal{K}}
\newcommand{\ECal}{\mathcal{E}}
\newcommand{\LMap}{\hat{\mathcal{L}}}
\newcommand{\TCal}{\mathcal{T}}
\newcommand{\UCal}{\mathcal{U}}
\newcommand{\StiefelProd}{\mathcal{M}}
\newcommand{\Lin}{\mathsf{L}}
\newcommand{\Pos}{\mathsf{Pos}}

\newcommand{\CPTNI}{\mathsf{CPTNI}}
\newcommand{\Tp}{\mathsf{T}}
\newcommand{\opt}{\mathrm{opt}}
\newcommand{\Fe}{F_{\mathrm{e}}}
\newcommand{\Fch}{F_{\mathrm{ch}}}
\newcommand{\Fopt}{F^{\opt}}
\newcommand{\Fopto}{F^{\opt}_{1}}
\newcommand{\eps}{\varepsilon}
\newcommand{\epsopt}{\varepsilon^{\opt}}
\newcommand{\Deltaopt}{\Delta^{\opt}}
\newcommand{\smin}[1]{\sigma_{\min+}(#1)}
\newcommand{\kgap}{\kappa^{(2)}}
\newcommand{\rC}{r_{\CCal}}

\newcommand{\Ktilde}{\widetilde K}
\newcommand{\PiRho}{\Pi_{\rho}}
\newcommand{\tracelesspart}[1]{\left(#1\right)^{\circ}}
\newcommand{\ket}[1]{\lvert #1\rangle}
\newcommand{\bra}[1]{\langle #1\rvert}
\newcommand{\kb}[1]{\lvert #1\rangle\!\langle #1\rvert}
\newcommand{\dket}[1]{\lvert #1\rangle\!\rangle}
\newcommand{\dbra}[1]{\langle\!\langle #1\rvert}
\newcommand{\dnorm}[1]{\lVert #1\rVert_{\diamond}}
\newcommand{\opnorm}[1]{\lVert #1\rVert_{\infty}}
\renewcommand{\Re}{\operatorname{Re}}
\providecommand{\tr}{}\renewcommand{\tr}{\operatorname{tr}}
\newcommand{\trL}{\tr_{L}}
\newcommand{\trC}{\tr_{C}}
\newcommand{\trS}{\tr_{S}}
\providecommand{\supp}{}\renewcommand{\supp}{\operatorname{supp}}
\providecommand{\rank}{}\renewcommand{\rank}{\operatorname{rank}}
\newcommand{\Ran}[1]{\operatorname{Ran}(#1)}
\newcommand{\id}{\mathrm{id}}

\theoremstyle{plain}
\newtheorem{theorem}{Theorem}
\newtheorem{corollary}{Corollary}
\newtheorem{lemma}{Lemma}
\theoremstyle{remark}

\newcommand{\startsupplement}{}
\newcommand{\restatedname}{}
\theoremstyle{plain}
\newtheorem*{restated}{\restatedname}

\begin{document}

\title{High-Rank Encoding Can Improve Approximate Quantum Error Correction}

\author{Bikun Li}
\email{bikunli@uchicago.edu}
\affiliation{Pritzker School of Molecular Engineering, University of Chicago,
Chicago, Illinois 60637, USA}
\author{Liang Jiang}
\email{liangjiang@uchicago.edu}
\affiliation{Pritzker School of Molecular Engineering, University of Chicago,
Chicago, Illinois 60637, USA}

\begin{abstract}
Conventional quantum-code constructions encode pure logical states as pure
code states, but this restriction can sacrifice performance. We show that
intrinsic encoding randomness can improve optimal entanglement fidelity. We
bound the loss from imposing a rank-one encoder and prove it is at most
quadratic near perfect recovery after joint optimization. The optimized
advantage survives small noise perturbations. An explicit noise family
requires higher-rank encoders arbitrarily close to perfect recovery, with
every optimal encoder mapping pure inputs to mixed code states.
\end{abstract}

\maketitle

Quantum error correction~(QEC) protects logical information by encoding it in
a physical system and applying a recovery after noise. Perfect correction
restores every encoded state exactly, while approximate quantum error
correction~(AQEC) seeks the best attainable performance when perfect recovery
is unavailable. AQEC theory provides approximate correctability criteria for
subspace and subsystem encodings, together with structured near-optimal
recoveries~\cite{Leung1997,Beny2010,Mandayam2012,Li2025,NgMandayam2010,
HaydenPenington2020}.
These results do not, however, determine how general an encoder must be to
achieve optimal approximate correction.

Conventional encoder constructions use a partial isometry, which maps an
orthonormal logical basis to orthonormal physical code states, as in
stabilizer~\cite{Shor1995,Gottesman1996,Terhal2015},
nonadditive~\cite{Rains1997,GrasslRoetteler2013}, and bosonic
codes~\cite{Chuang1997,Terhal2020Bosonic}. AQEC also permits nonorthogonal
physical representatives of logical basis states, including finite-energy
Gottesman--Kitaev--Preskill codewords and finite-amplitude coherent-state
representatives of cat qubits~\cite{Gottesman2001,Grimsmo2021,Jafarzadeh2025,
Mirrahimi2014,Wang2026}. Such encodings can still be realized with a single
Kraus operator and therefore do not establish whether optimal approximate
correction can require a higher-rank encoder. More generally, a rank-one
completely positive trace-non-increasing~(CPTNI) encoder need not preserve
trace, whereas every rank-one completely positive trace preserving~(CPTP)
encoder is a partial-isometry encoder.

This rank question has one precedent in probabilistic quantum error
correction~(pQEC). The pQEC framework demands exact recovery on a postselected
branch and maximizes the corresponding success probability. Under this
restriction, higher-rank encoders may strictly outperform rank-one
encoders~\cite{Kukulski2023}. This raises the broader question of whether
unrestricted QEC schemes without an exact-recovery constraint can likewise
benefit from higher-rank encoding. Using entanglement fidelity as the objective,
we answer affirmatively and show that a higher-rank encoder can be necessary to
attain the global optimum.

Establishing this necessity is subtle. Higher-rank encoders already occur in
operator QEC when a mixed gauge state is prepared, but this does not establish
necessity because a pure gauge choice protects the same
subsystem~\cite{Kribs2005,Kribs2006,Lidar2014}. Joint encoder--decoder
optimization creates a separate difficulty. Fixing either map gives a
semidefinite program~(SDP), whereas the joint problem is biconvex and
alternating optimization has no global
guarantee~\cite{Reimpell2005,Fletcher2007,Kosut2008,Kosut2009}. Earlier
searches found rank-one optima without establishing universal
optimality~\cite{Reimpell2005}. A proof of necessity must therefore compare the
global unrestricted and rank-one optima.

At perfect recovery, Barnum, Knill, and Nielsen showed that the encoder may be
chosen as a maximal partial isometry~\cite{Barnum2000}. Because the encoding
space is assumed below to be at least as large as the logical space, this
encoder may also be taken trace preserving~(TP). Thus rank-one encoding loses
nothing for CPTP schemes at exact correction. Their result in the near-perfect
regime is not tight. For any fixed CPTP encoder, noise map, and decoder, their
result provides a partial-isometry encoder whose scheme fidelity under the same
noise map and decoder is at least the square of the original fidelity.

We strengthen and extend this benchmark by proving a state-dependent upper
bound on the infidelity increase incurred when an arbitrary CPTNI encoder is
replaced by a partial-isometry encoder while the noise map and decoder remain
fixed. Compared with the known estimate~\cite{Barnum2000}, the bound is tighter
in the regions of small and large infidelity and is quadratic near perfect
recovery. The same bound controls the fully optimized fidelity gap between
unrestricted and rank-one encoders, and any positive optimized gap persists
under sufficiently small perturbations of the noise map.

We then construct an explicit noise family whose optimized gap is strictly
positive for every nonzero parameter, including parameters arbitrarily close
to perfect recovery.
Throughout this regime, every optimal encoder has Choi matrix rank greater
than one, maps every pure logical input to a mixed state, and has nonvanishing
entropy exchange at the maximally mixed input. The encoder randomness is therefore
unavoidable at the optimum, and forbidding it strictly lowers the best
fidelity.

Related benefits from randomness occur in other quantum-information tasks.
Noisy preprocessing raises secret-key
rates~\cite{KrausGisinRenner2005,RenesSmith2007}, while shared randomness
improves fixed entanglement-purification protocols under source
uncertainty~\cite{Zang2026}. Those gains exploit task-specific nonlinear
responses to randomization. Here, for fixed noise and decoder, the fidelity
objective is linear in the encoder's Choi matrix, so a classical mixture of
encoder quantum channels cannot outperform its best constituent under that
decoder. The advantage instead arises from the intrinsic randomness of a
high-rank encoder.

The following sections introduce the setting, establish the general upper
bound and the stability of the gap under noise perturbations, and present the
high-rank family and the asymptotic coefficient of its optimized gap.

\section{Setting}\label{sec:setting}

Let $\HL=\Cbb^{d_L}$ and $\HC=\Cbb^{d_C}$ be the logical and
code spaces, with $1<d_L\leq d_C<\infty$. For a finite-dimensional Hilbert
space $\HH$, write $\Lin(\HH)$ for the space of linear operators on $\HH$.
The encoder
$\CCal:\Lin(\HL)\to\Lin(\HC)$, noise map
$\NN:\Lin(\HC)\to\Lin(\HC)$, and decoder
$\DD:\Lin(\HC)\to\Lin(\HL)$ are CPTNI maps, and the QEC scheme is the
composite map $\Phi:=\DD\circ\NN\circ\CCal$. CPTNI maps describe postselected
outcomes of quantum instruments, whereas CPTP maps form the deterministic
special case known as quantum channels. Under input-first vectorization, an
operator $M:\HH_1\to\HH_2$ is represented by
$\dket{M}\in\HH_{1'}\otimes\HH_2$, where $\HH_{1'}$ is a
reference copy of $\HH_1$. For $\HH_1=\HL$, we denote this copy by $\HRef$.
For a map $\Lambda:\Lin(\HH_1)\to\Lin(\HH_2)$ with Kraus operators $L_j$,
its input-first Choi matrix is
$J(\Lambda):=\sum_j\dket{L_j}\dbra{L_j}$~\cite{Choi1975,Jamiolkowski1972}.
The Choi matrix rank is the minimum number
of Kraus operators, and we write $\rC:=\rank J(\CCal)$. A rank-one CPTNI encoder
has the form $V(\cdot)V^{\dagger}$, with $V:\HL\to\HC$ and
$V^{\dagger}V\preceq I_L$. It is CPTP exactly when
$V^{\dagger}V=I_L$.

For a state $\rho$ on $\HL$, let $\Hr:=\supp(\rho)$ and let
$\PiRho$ be its support projector. Its canonical purification is
$\ket{\psi_{\rho}}:=\dket{\sqrt{\rho}}\in\HRef\otimes\HL$. If $F_j$ are
Kraus operators of $\Phi$, the entanglement fidelity~\cite{Schumacher1996} is
\begin{equation}\label{eq:entanglement-fidelity}
\begin{aligned}
\Fe(\rho,\Phi)
&:=\bra{\psi_{\rho}}(\id_{L'}\otimes\Phi)
(\kb{\psi_{\rho}})\ket{\psi_{\rho}}\\
&=\sum_j\left\lvert\tr(\rho F_j)\right\rvert^2.
\end{aligned}
\end{equation}
Entanglement fidelity probes only errors visible to $\rho$, so it need not
determine the entire logical map. For full rank $\rho$,
$\Fe(\rho,\DD\circ\NN\circ\CCal)=1$ forces the composite map to be the
identity, $\DD\circ\NN\circ\CCal=\id_L$, for which a partial-isometry encoder
suffices~\cite{Knill1997,Barnum2000}. When
$\DD\circ\NN\circ\CCal=p_{\mathrm{succ}}\id_L$ is imposed, as in pQEC,
$\Fe(\rho,\DD\circ\NN\circ\CCal)=p_{\mathrm{succ}}$ independently of
$\rho$. For $p_{\mathrm{succ}}\in(0,1)$, pQEC can require a high-rank
encoder~\cite{Kukulski2023}.

\section{Mind the Gap}\label{sec:theory}

We sharpen the near-perfect bound developed in Ref.~\cite{Barnum2000} by
deriving a quadratic small-infidelity estimate and extending it to general
CPTNI maps. The refined bound depends explicitly on the quantum state $\rho$
through $s:=\smin{\rho}\in(0,1]$, its smallest positive eigenvalue.
For $0\leq x\leq s$, set
\begin{equation}\label{eq:ell-definition}
\ell_s(x):=\frac{x}{1+\sqrt{1-x/s}}.
\end{equation}
The function $\ell_s$ is strictly increasing and therefore invertible on
$[0,s]$, with $0\leq\ell_s(x)\leq x$.

\begin{theorem}[improved upper bound]
\label{thm:quadratic}
Let $\rho\succeq0$ on $\HL$ satisfy $\tr\rho=1$, and let
$\ACal:=\DD\circ\NN$ and $\CCal$ be CPTNI maps. Denote
\begin{equation}\label{eq:scheme-infidelity}
\eps:=1-\Fe(\rho,\ACal\circ\CCal).
\end{equation}
A partial-isometry encoder $\CCal_1$ can then be chosen so that, with
\begin{equation}\label{eq:rank-one-infidelity}
\eps_1:=1-\Fe(\rho,\ACal\circ\CCal_1),
\end{equation}
the corresponding infidelity gap satisfies, for every $0\leq\eps\leq1$,
\begin{equation}\label{eq:specified-gap-bound}
\Delta:=\eps_1-\eps\leq
\begin{cases}
(s^{-1}-1)\ell_s(\eps)^2,
&0\leq\eps\leq s,\\
(1-\eps)\min\{\eps,1-s\},
&s<\eps\leq1
\end{cases}.
\end{equation}
\end{theorem}

\begin{proof}
At $\eps=1$, any encoder gives $\eps_1\leq1$, so $\Delta$ satisfies
Eq.~\eqref{eq:specified-gap-bound}. Suppose henceforth that $\eps<1$. Choose Kraus
operators $C_a$ and $A_i$ for $\CCal$ and $\ACal$, respectively. For each
encoder branch, let
$w_a:=\tr(\rho C_a^{\dagger}C_a)$ be its input weight and
$f_a:=\sum_i|\tr(\rho A_iC_a)|^2$ its fidelity contribution. The
Cauchy--Schwarz inequality
and $\sum_iA_i^{\dagger}A_i\preceq I_C$ give
$0\leq f_a\leq\tr[\rho C_a^{\dagger}(\sum_iA_i^{\dagger}A_i)C_a]\leq w_a$.
For $w_a>0$, let
$\eps_a:=1-f_a/w_a$ be the branch's relative defect. Trace non-increase and the
Kraus representation formula give $w:=\sum_a w_a\leq1$ and
$\sum_a f_a=1-\eps$. Since $\eps<1$ and $f_a\leq w_a$, we have $w>0$. Hence the
average of $\eps_a$ with normalized weights $w_a/w$ is
$[w-(1-\eps)]/w\leq\eps$. Choose a positive-weight branch satisfying
$\eps_a\leq\eps<1$ and set $x:=\eps_a$.

We next work on $\Hr$. Take the polar decomposition $C_a\PiRho=VP$. Since
$d_C\geq d_L$, complete the polar factor
on $\ker P$ to a partial isometry $V:\Hr\to\HC$ with initial space $\Hr$,
without changing $VP$. Thus
$V^{\dagger}V=\PiRho$. Normalize $P$ as $\widetilde P:=P/\sqrt{w_a}$. With
$K_i:=\PiRho A_iV:\Hr\to\Hr$, trace non-increase gives
$\sum_iK_i^{\dagger}K_i\preceq\PiRho$, while the support of $\rho$ gives
$\tr(\rho\widetilde P^2)=1$ and
$\sum_i|\tr(\rho K_i\widetilde P)|^2=1-x$. Since $x\leq\eps<1$, the coefficients
$v_i:=\tr(\rho K_i\widetilde P)^*/\sqrt{1-x}$ therefore satisfy
$\sum_i|v_i|^2=1$, so $v$ is a unit vector. Using these coefficients, set
$\Ktilde:=\sum_iv_iK_i$. By construction,
$\tr(\rho\Ktilde\widetilde P)=\sqrt{1-x}$. The Cauchy--Schwarz inequality
gives $\Ktilde^{\dagger}\Ktilde\preceq\sum_iK_i^{\dagger}K_i
\preceq\PiRho$, so $\Ktilde:\Hr\to\Hr$ is a contraction. Take its full
polar decomposition $\Ktilde=HU^{\dagger}$, where $H:\Hr\to\Hr$ satisfies
$0\preceq H\preceq\PiRho$ and $U:\Hr\to\Hr$ is unitary. The
Cauchy--Schwarz inequality gives
\begin{equation}\label{eq:second-moment}
\begin{aligned}
1-x&=|\tr(\rho HU^{\dagger}\widetilde P)|^2\\
&\leq\tr(\rho H^2)\tr(\rho\widetilde PUU^{\dagger}\widetilde P)
=\tr(\rho H^2).
\end{aligned}
\end{equation}

Set $W:=VU\PiRho$. Since $W^{\dagger}W=\PiRho$, it defines the
partial-isometry encoder $\CCal_1(\omega):=W\omega W^{\dagger}$.
The Kraus representation formula and the Cauchy--Schwarz inequality now give
\begin{equation}\label{eq:rounded-fidelity}
\begin{aligned}
1-\eps_1
&=\sum_i\left|\tr(\rho A_iW)\right|^2
=\sum_i\left|\tr(\rho A_iVU)\right|^2\\
&=\sum_i\left|\tr(\rho K_iU)\right|^2
\geq\left|\tr(\rho\Ktilde U)\right|^2\\
&=\left[\tr(\rho H)\right]^2.
\end{aligned}
\end{equation}

First suppose $0\leq\eps\leq s$. To turn the second-moment estimate into a
first-moment bound, write $y:=1-\tr(\rho H)$. Because
$\PiRho-H\preceq\PiRho-H^2$, we have $0\leq y\leq x\leq\eps\leq s$.
Positivity gives
$(\PiRho-H)^2\preceq(\PiRho-H)\tr(\PiRho-H)$, while
$\rho\succeq s\PiRho$ gives $\tr(\PiRho-H)\leq y/s$. Hence
$\tr[\rho(\PiRho-H)^2]\leq y^2/s$, and
$x\geq\tr[\rho(\PiRho-H^2)]\geq2y-y^2/s\equiv\ell_s^{-1}(y)$. Since
$\ell_s^{-1}$ is increasing on $[0,s]$, we obtain $y\leq\ell_s(x)$, or
equivalently $\tr(\rho H)\geq1-\ell_s(x)$. Substituting this bound into
Eq.~\eqref{eq:rounded-fidelity} yields
\begin{equation}
\begin{aligned}
1-\eps_1
&\geq[1-\ell_s(x)]^2
=1-x-(s^{-1}-1)\ell_s(x)^2\\
&\geq1-\eps-(s^{-1}-1)\ell_s(\eps)^2.
\end{aligned}
\end{equation}
The equality uses the identity implied by Eq.~\eqref{eq:ell-definition}, and
the last line follows from $x\leq\eps$ and monotonicity of $\ell_s$.
Rearranging proves the first branch of Eq.~\eqref{eq:specified-gap-bound}.

Now suppose $s<\eps<1$. From Eq.~\eqref{eq:second-moment}, we obtain
$\tr(\rho H^2)\geq1-x>0$. Since $H\preceq\PiRho$, one has
$0<\opnorm{H}\leq1$ and $H^2\preceq\opnorm{H}H$. Using
$\rho\succeq s\PiRho$ and $\tr(H^2)\geq\opnorm{H}^2$ gives
\begin{equation}
s\opnorm{H}^2\leq\tr(\rho H^2)
\leq\opnorm{H}\tr(\rho H).
\end{equation}
Together with Eq.~\eqref{eq:second-moment} and $x\leq\eps$, this gives
\begin{equation}
\tr(\rho H)\geq\frac{\tr(\rho H^2)}{\opnorm{H}}
\geq\sqrt{s\tr(\rho H^2)}
\geq\sqrt{s(1-\eps)}.
\end{equation}
Moreover, $H^2\preceq H$, Eq.~\eqref{eq:second-moment}, and $x\leq\eps$ give
$\tr(\rho H)\geq\tr(\rho H^2)\geq1-\eps$. Combining these two estimates
with Eq.~\eqref{eq:rounded-fidelity} yields
\begin{equation}
1-\eps_1\geq(1-\eps)\max\{1-\eps,s\}.
\end{equation}
Therefore, $\Delta\leq(1-\eps)\min\{\eps,1-s\}$, which proves the second
branch for $s<\eps<1$.
\end{proof}

Theorem~\ref{thm:quadratic} remains valid when $\ACal$, $\CCal$, and
$\CCal_1$ are restricted to CPTP maps. The proof is analogous and omitted.
We denote the piecewise right-hand side of Eq.~\eqref{eq:specified-gap-bound} by
$\overline{\Delta}(\eps,s)$. Fig.~\ref{fig:gap}(a) compares this upper bound
with the bound $\Delta\leq\eps(1-\eps)$ implicit in the proof of
Ref.~\cite{Barnum2000}.

Without assuming optimality,
Eqs.~\eqref{eq:scheme-infidelity}--\eqref{eq:specified-gap-bound} compare specified
encoders under a fixed decoder. They therefore do not capture
the intrinsic limitation of rank-one encoding for a given $\rho$ and noise map
$\NN$. For the joint encoder--decoder optimization, introduce
\begin{align}
\Fopt(\rho,\NN)&:=\max_{\CCal,\DD\in\CPTNI}
\Fe(\rho,\DD\circ\NN\circ\CCal),
\label{eq:joint-optimum}\\
\Fopto(\rho,\NN)&:=
\max_{\substack{\CCal,\DD\in\CPTNI\\ \rC\leq1}}
\Fe(\rho,\DD\circ\NN\circ\CCal),
\label{eq:joint-rank-one-optimum}\\
\epsopt(\rho,\NN)&:=1-\Fopt(\rho,\NN),
\label{eq:joint-infidelity}\\
\Deltaopt(\rho,\NN)&:=\Fopt(\rho,\NN)-\Fopto(\rho,\NN).
\label{eq:joint-gap}
\end{align}
Both fidelities are now fully optimized for the given $\rho$ and $\NN$. Hence
$\Deltaopt>0$ certifies that every rank-one encoder falls short of the
unrestricted optimum.
\begin{corollary}[optimized upper bound]
\label{cor:optimized-ceiling}
The optimized gap satisfies
\begin{equation}\label{eq:optimized-ceiling}
0\leq\Deltaopt(\rho,\NN)
\leq\overline{\Delta}\big(\epsopt(\rho,\NN),s\big).
\end{equation}
\end{corollary}

\begin{proof}
Let $(\CCal^{\opt},\DD^{\opt})$ attain $\Fopt(\rho,\NN)$, set
$\ACal^{\opt}:=\DD^{\opt}\circ\NN$, and let $\CCal_1$ be the choice in
Theorem~\ref{thm:quadratic} for $(\ACal^{\opt},\CCal^{\opt})$. It gives a
fixed-pair loss $\Delta:=\eps_1-\epsopt(\rho,\NN)$ no larger than
$\overline{\Delta}(\epsopt(\rho,\NN),s)$. The pair
$(\CCal_1,\DD^{\opt})$ could be suboptimal, so
$0\leq\Deltaopt(\rho,\NN)\leq\Delta$. Combining the inequalities proves
Eq.~\eqref{eq:optimized-ceiling}.
\end{proof}

\begin{theorem}[Lipschitz stability]
\label{thm:gap-stability}
For fixed $\rho$ and CPTNI maps $\NN$ and $\NN'$, the optimized gap obeys
\begin{equation}\label{eq:gap-stability}
|\Deltaopt(\rho,\NN)-\Deltaopt(\rho,\NN')|
\leq2\dnorm{\NN-\NN'}.
\end{equation}
If $\NN$ and $\NN'$ are CPTP maps, the factor of two can be omitted.
\end{theorem}

\emph{Proof sketch.}
For each fixed feasible encoder--decoder pair, entanglement fidelity is
$1$-Lipschitz with respect to the noise map. Passing this estimate to the
unrestricted and rank-one optima and applying the triangle inequality yield
the $2$-Lipschitz bound in
Eq.~\eqref{eq:gap-stability}. For CPTP noise, a TP
completion of the decoder makes the difference between the relevant Choi
matrices traceless, so the variational characterization of trace distance
halves the fidelity constants. See the Supplemental Material~(SM)~\cite{SM}.

Theorem~\ref{thm:gap-stability} shows that a positive optimized gap
persists under sufficiently small diamond-norm perturbations within the CPTNI
class. In particular, when $\Deltaopt>0$, the advantage of allowing higher-rank
encoders in the entanglement-fidelity optimization survives small changes to
the noise map. It is therefore not confined to a singular noise map, nor is it
a fine-tuned artifact. Since
$\ell_s(x)=x/2+\OO(x^2)$, the small-infidelity ceilings in
Eqs.~\eqref{eq:specified-gap-bound} and~\eqref{eq:optimized-ceiling} have
quadratic coefficient $(s^{-1}-1)/4$. The actual $\Deltaopt$ may be smaller or
vanish, so this coefficient is only a ceiling. For a family approaching exact
correction, we therefore study the asymptotic quadratic coefficient
$\kgap:=\lim_{\epsopt\to0^+}\Deltaopt/(\epsopt)^2$ whenever this limit
exists.

\begin{figure}[!t]
\centering
\includegraphics[width=\columnwidth]{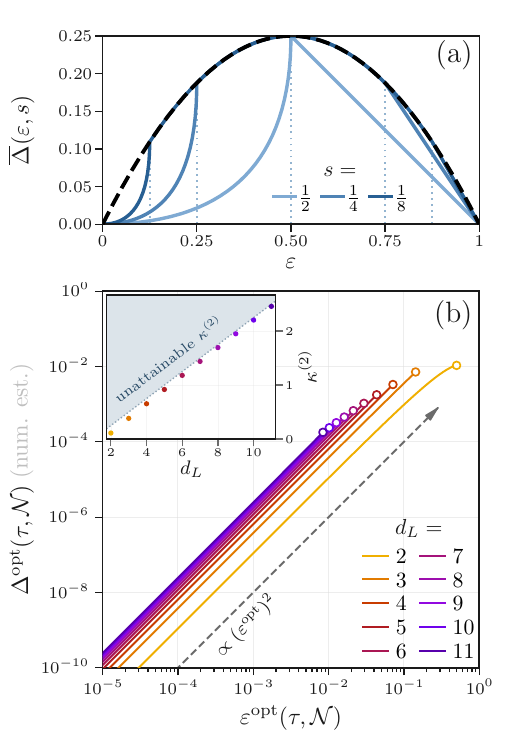}
\caption{(a) The upper bound $\overline{\Delta}(\eps,s)$ given by
Eq.~\eqref{eq:specified-gap-bound} for different values of $s$. The black dashed
curve is the earlier $\eps(1-\eps)$ bound derived from
Ref.~\cite{Barnum2000}.
(b) Numerical estimates of the optimized gap versus optimized infidelity for
the noise maps constructed below and different values of $d_L$, with
$d_C=d_L^2$. The inset compares the finite-error ratios of the plotted gaps to
$(\epsopt)^2$ with the asymptotic ceiling $(d_L-1)/4$. The
exact limiting coefficient is given by Eq.~\eqref{eq:example-kappa}.}
\label{fig:gap}
\end{figure}

\section{A High-Rank Example}\label{sec:example}

We provide a concrete family that encodes one logical qudit into two physical
qudits and whose $\kgap$ approaches the aforementioned upper bound as $d_L$
grows. Detailed derivations of the construction and its optimality properties
are given in the SM~\cite{SM}. Set
$\HS=\Cbb^{d_L}$ and $\HC=\HL\otimes\HS$, so $d_C=d_L^2$, and use the
maximally mixed input $\tau:=I_L/d_L$. For this input, $s=1/d_L$ and the
ceiling coefficient in Eq.~\eqref{eq:optimized-ceiling} is $(d_L-1)/4$.
Let $\ket{\psi_\tau}:=d_L^{-1/2}\sum_{i=0}^{d_L-1}\ket{i}_L\otimes\ket{i}_S\in\HC$
be maximally entangled.
For $0\leq p<1$, set
\begin{equation*}
\begin{gathered}
\lambda_1:=1-\frac{p}{d_L^2-2},\quad
\lambda_2:=\frac{(1-\lambda_1)(d_L^2\lambda_1-1)}
{d_L^2\lambda_4},\\
\lambda_3:=\frac{2(1-\lambda_1)^2}{d_L^2\lambda_4},\quad
\lambda_4:=\frac{d_L^2-1+p}{d_L^2},\quad
\lambda_5:=\frac{1-p}{d_L^2}.
\end{gathered}
\end{equation*}
They obey $\lambda_1\geq\lambda_2\geq\lambda_3\geq0$ and
$\lambda_1+\lambda_2+(d_L^2-1)\lambda_3
=\lambda_4+\lambda_5=1$.
Set $\hat{\pi}_C=\kb{\psi_\tau}_{L\otimes S}$ and define
\begin{equation}\label{eq:example-generator}
Q_p:=
\sqrt{d_L}\left[
\sqrt{\frac{\lambda_4}{d_L^2-1}}(I_C-\hat{\pi}_C)
+\sqrt{\lambda_5}\hat{\pi}_C\right].
\end{equation}
Before defining the noise map, consider the following TP encoder and decoder:
\begin{equation}\label{eq:example-pair}
\CCal_p^{\opt}(\rho):=Q_p(\rho\otimes I_S)Q_p,
\qquad \DD^{\opt}:=\trS,
\end{equation}
where $\rho\in\Lin(\HL)$ and $\trS$ denotes the partial trace over $\HS$.
Since $Q_p\succ0$, $\CCal_p^{\opt}(\rho)$ has rank $d_L$ whenever $\rho$ is
pure, so this encoder maps every pure logical state to a mixed code state.
Write the Kraus operators of $\CCal_p^{\opt}$ and $\DD^{\opt}$ as
$C_i(p):=Q_p(I_L\otimes\ket{i}_S)$ and
$D_j:=I_L\otimes\bra{j}_S$. Set $B_{ij}(p):=C_i(p)D_j$, let $P_p$ be the
orthogonal projector onto
$\operatorname{span}\{\dket{B_{ij}(p)}\}_{i,j=0}^{d_L-1}
\subseteq\HC\otimes\HC$, and write
$P_p^{\perp}:=I_{C\otimes C}-P_p$. Let $X_p$ be given by
\begin{equation}\label{eq:example-xp}
X_p=\lambda_3 I_C+
\frac{\lambda_2}{\lambda_4}\hat{\pi}_C.
\end{equation}
The noise map $\NN_p$ is specified through the Choi matrix
\begin{equation}\label{eq:example-noise-main}
J(\NN_p^{\dagger})=\lambda_1P_p
+P_p^{\perp}(I_C\otimes X_p)P_p^{\perp},
\end{equation}
which defines a unital and covariant quantum channel.
To express its action explicitly, write
$\ket{ij}_C:=\ket{i}_L\otimes\ket{j}_S$ and introduce the Hilbert--Schmidt
orthonormal operators
\begin{equation}\label{eq:example-eij}
E_{ij}(p):=
\frac{\ket{ij}_C\!\bra{\psi_\tau}-\sqrt{\lambda_5}B_{ij}(p)^{\dagger}}
{\sqrt{\lambda_4}}.
\end{equation}
For every $\omega\in\Lin(\HC)$, the noise channel $\NN_p$ acts as follows:
\begin{equation}\label{eq:example-noise-action}
\begin{aligned}
\NN_p(\omega)={}&\lambda_3\tr(\omega)I_C
+(\lambda_1-\lambda_3)
\trS(Q_p\omega Q_p)\otimes I_S\\
&+\lambda_2\sum_{i,j=0}^{d_L-1}
E_{ij}(p)\omega E_{ij}(p)^{\dagger}.
\end{aligned}
\end{equation}
For any CPTNI encoder--decoder pair, set
$\Phi=\DD\circ\NN_p\circ\CCal$, $\BCal=\CCal\circ\DD$, and
$Z_p:=\lambda_1I_{C\otimes C}-J(\NN_p^{\dagger})\succeq0$. Then
\begin{equation}\label{eq:example-certificate}
\Fe(\tau,\Phi)=\frac{\lambda_1}{d_L^2}\tr J(\BCal)-\frac{1}{d_L^2}\tr[Z_pJ(\BCal)]\leq\lambda_1.
\end{equation}
The pair $(\CCal_p^{\opt},\DD^{\opt})$ saturates this bound throughout
$0\leq p<1$,
yielding
$\Fopt(\tau,\NN_p)=\lambda_1$ and
$\epsopt(\tau,\NN_p)=1-\lambda_1=p/(d_L^2-2)$. Equality forces
$\BCal$ to be TP and
$\supp(J(\BCal))\subseteq\ker Z_p$.
For $p>0$, the TP constraint within this support uniquely fixes
$\BCal$ to $\BCal_p:=\CCal_p^{\opt}\circ\DD^{\opt}$. Since
$\BCal_p(I_C)=d_L Q_p^2$ has rank $d_L^2$, whereas
$\rank[\BCal(I_C)]\leq d_L r_{\CCal}$, every encoder in an unrestricted
optimal pair has $r_{\CCal}\geq d_L$ for $p>0$, and thus
$\Deltaopt(\tau,\NN_p)>0$.
For $d_L=2$, this family is related to a pQEC construction in
Ref.~\cite{Kukulski2023}, but our decoder averages over all outcome branches
and our optimization imposes no pQEC constraint.

For $p\in(0,1)$, every encoder in an unrestricted optimal pair is of the form
$\CCal_p^{\opt}\circ\UCal$ for a unitary channel $\UCal$.
Since $\tau$ is invariant under $\UCal$ and
$\tr[\tau C_i(p)^{\dagger}C_j(p)]=\delta_{ij}/d_L$, every such encoder has
entropy exchange $\log_2 d_L$ on $\tau$. Thus high-rank encoding with
nonvanishing randomness remains necessary as $p\to0^+$, i.e., $\epsopt\to0^+$.
Remarkably, $\{C_i(p)^{\dagger}C_j(p)\}_{i,j}$ is found to be a linearly
independent set for $p\in(0,1)$, so $\CCal_p^{\opt}$ sits at an extreme point
of the CPTP set according to Choi's extremality criterion for quantum
channels~\cite{Choi1975}. In other words, such a high-rank encoder cannot be
decomposed as a nontrivial mixture of quantum channels. We expect that, in
more general cases where high-rank encoding prevails, the optimal encoder can
be chosen as an extreme point of the CPTP set, as no perturbation of the
encoder will improve the objective.

For the family $(\tau,\NN_p)$ constructed in this section, the
SM~\cite{SM} derives a tangent-space lower bound for the $\rC=1$
encoder--decoder optimization.
When applied to the unrestricted optimal scheme
$(\CCal_p^{\opt},\DD^{\opt})$, the
construction in the proof of Theorem~\ref{thm:quadratic} produces a feasible
rank-one encoder path that saturates this bound through quadratic order.
Consequently, the asymptotic quadratic gap coefficient for this model is
\begin{equation}\label{eq:example-kappa}
\kgap=\frac{d_L-1}{4}\left(\frac{d_L^2-2}{d_L^2-1}\right)^2.
\end{equation}
This analysis does not determine higher-order corrections to the optimized gap.
Fig.~\ref{fig:gap}(b) provides numerical estimates of $\Deltaopt$ at finite $p$.

\section*{Summary and outlook}

By constructing a new rank-one encoder, we derived a tighter general quadratic
upper bound on the optimized fidelity advantage of unrestricted over rank-one
encoders near perfect recovery and proved that any positive advantage persists
under sufficiently small noise perturbations. Our explicit family realizes
this advantage arbitrarily close to perfect recovery
while requiring genuinely high-rank optimal encoders with nonvanishing
randomness. The exact quadratic coefficient for this family approaches the
upper bound as the logical dimension grows.

Future work could identify practical criteria for when noise favors
high-rank encoders, extend the analysis to infinite-dimensional systems,
especially continuous-variable settings, and explore quantum resource-theoretic
interpretations of encoder Choi matrix rank and its operational advantage.

\section*{Acknowledgement}

B.L.\ thanks Yuqi Li, Mahadevan Subramanian, and Yat Wong for useful
discussions and comments. We acknowledge the use of AI for technical discussions
and some assistance in prose. We acknowledge support from the U.S.\ Army
Research Office~(ARO)
under Award No.~W911NF-23-1-0077, the ARO Multidisciplinary University Research
Initiative~(MURI) under Award No.~W911NF-21-1-0325, the Air Force Office of
Scientific Research~(AFOSR) MURI under Award Nos.~FA9550-21-1-0209 and
FA9550-23-1-0338, the Office of Naval Research~(ONR) MURI under Award
No.~N000142612102, the Defense Advanced Research Projects Agency~(DARPA) under
Award No.~HR0011-24-9-0361, and the National Science Foundation~(NSF) under
Award Nos.~ERC-1941583, OMA-2137642, OSI-2326767, CCF-2312755, and
OSI-2426975.
This material is based upon work supported by the U.S.\ Department of Energy,
Office of Science, National Quantum Information Science Research Centers and
Advanced Scientific Computing Research~(ASCR) program under contract number
DE-AC02-06CH11357 as part of the InterQnet quantum networking project.

\clearpage
\onecolumngrid
\newpage
\startsupplement

\begin{center}
{\large\textbf{Supplemental Material for ``High-Rank Encoding Can Improve
Approximate Quantum Error Correction''}}\\[10pt]
Bikun Li and Liang Jiang\\[2pt]
{\itshape Pritzker School of Molecular Engineering, University of Chicago,
Chicago, Illinois 60637, USA}
\end{center}

\setcounter{equation}{0}
\setcounter{figure}{0}
\setcounter{table}{0}
\setcounter{section}{0}
\setcounter{lemma}{0}
\setcounter{corollary}{0}
\renewcommand{\theequation}{S\arabic{equation}}
\renewcommand{\theHequation}{S\arabic{equation}}
\renewcommand{\thefigure}{S\arabic{figure}}
\renewcommand{\thetable}{S\arabic{table}}
\renewcommand{\thesection}{S\arabic{section}}
\renewcommand{\thesubsection}{\thesection.\Alph{subsection}}
\renewcommand{\thelemma}{S\arabic{lemma}}
\renewcommand{\theHlemma}{S\arabic{lemma}}
\renewcommand{\thecorollary}{S\arabic{corollary}}
\renewcommand{\theHcorollary}{S\arabic{corollary}}
\setcounter{secnumdepth}{2}

\addtocontents{toc}{\protect\smtocstart}
\supplementtableofcontents
\medskip

\section{Conventions and structural facts}\label{sec:sm-structure}

For a finite-dimensional Hilbert space $\HH$, let
$\Pos(\HH):=\{X\in\Lin(\HH):X\succeq0\}$ denote its positive-semidefinite
cone.

With the input-first vectorization convention of the main text, the canonical
purification obeys, for $\rho,K\in\Lin(\HL)$ with $\rho$ a state,
\begin{equation}\label{eq:canonical-check}
\tr_{L'}\dket{\sqrt\rho}\dbra{\sqrt\rho}=\rho,
\qquad
\dbra{\sqrt\rho}(I_{L'}\otimes K)\dket{\sqrt\rho}=\tr(\rho K).
\end{equation}
These identities give its normalization and the formula in terms of Kraus operators in
Eq.~\eqref{eq:entanglement-fidelity}.

For a linear map
$\Phi:\Lin(\HH_{\mathrm{in}})\to\Lin(\HH_{\mathrm{out}})$ and
$\omega\in\Pos(\HH_{\mathrm{in}})$, the dressed Choi matrix of $\Phi$ in
the chosen input basis is
\begin{equation}\label{eq:sm-dressed-choi}
J_{\omega}(\Phi):=
((\sqrt{\omega})^{\Tp}\otimes I_{\mathrm{out}})J(\Phi)
((\sqrt{\omega})^{\Tp}\otimes I_{\mathrm{out}}),
\end{equation}
where $(\cdot)^{\Tp}$ denotes transposition in that basis. If $\Phi$ is
completely positive trace-non-increasing~(CPTNI), then
$J_{\omega}(\Phi)\succeq0$ and
$\tr J_{\omega}(\Phi)\leq\tr\omega$, with equality in the trace bound when
$\Phi$ is trace preserving~(TP).

For full rank $\rho$, let the CPTNI map $\Phi=\DD\circ\NN\circ\CCal$ have
Kraus operators $L_i$. Applying the Cauchy--Schwarz inequality to the
Hilbert--Schmidt inner product with $X=\sqrt\rho$ and $Y=L_i\sqrt\rho$ and
using trace non-increase, we obtain
$\Fe(\rho,\Phi)\leq\sum_i\tr(\rho L_i^{\dagger}L_i)\leq1$.
Equality in the summed bound saturates each Cauchy--Schwarz inequality, so
$L_i\sqrt\rho=c_i\sqrt\rho$. Full rank then gives $L_i=c_iI_L$ and
$\Phi=\id_L$. The standard Knill--Laflamme conditions and
channel-reversibility results then yield the exact-QEC structure stated in the
main text, with every branch corrected by the same
decoder~\cite{Knill1997,Nielsen1998,Barnum2000}.

\section{Lipschitz stability}\label{sec:sm-stability}

\renewcommand{\restatedname}{Theorem~\ref{thm:gap-stability} (Lipschitz stability)}
\begin{restated}
For fixed $\rho$ and CPTNI maps $\NN$ and $\NN'$, the optimized gap obeys
\begin{equation}
|\Deltaopt(\rho,\NN)-\Deltaopt(\rho,\NN')|
\leq2\dnorm{\NN-\NN'}.
\end{equation}
If $\NN$ and $\NN'$ are completely positive trace preserving~(CPTP)
maps, the factor of two can be omitted.
\end{restated}

\begin{proof}
Fix a CPTNI encoder $\CCal$ and decoder $\DD$, and abbreviate
$f_{\CCal}(\NN):=\Fe(\rho,\DD\circ\NN\circ\CCal)$.
The dressed Choi matrix
$J_{\rho}(\DD\circ(\NN-\NN')\circ\CCal)$ is Hermitian but need not
be traceless. Since a
completely positive trace-non-increasing map has diamond norm at most one,
H\"older's inequality, submultiplicativity of the diamond norm, and its
definition give
\begin{equation}\label{eq:fixed-lipschitz}
\begin{aligned}
\left\lvert f_{\CCal}(\NN)-f_{\CCal}(\NN')\right\rvert
&=\left\lvert\tr\left[
\kb{\psi_\rho}J_{\rho}(\DD\circ(\NN-\NN')\circ\CCal)
\right]\right\rvert\\
&\leq
\left\|J_{\rho}(\DD\circ(\NN-\NN')\circ\CCal)\right\|_1
\leq\dnorm{\DD\circ(\NN-\NN')\circ\CCal}\\
&\leq\dnorm{\NN-\NN'}.
\end{aligned}
\end{equation}
The bound in Eq.~\eqref{eq:fixed-lipschitz} is uniform over all feasible pairs.
It therefore passes to the suprema over the unrestricted and rank-one feasible
sets, and interchanging $\NN$ and $\NN'$ gives the reverse inequalities. Thus
\begin{subequations}\label{eq:optimized-fidelity-lipschitz}
\begin{align}
|\Fopt(\rho,\NN)-\Fopt(\rho,\NN')|
&\leq\dnorm{\NN-\NN'},
\label{eq:unrestricted-fidelity-lipschitz}\\
|\Fopto(\rho,\NN)-\Fopto(\rho,\NN')|
&\leq\dnorm{\NN-\NN'},
\label{eq:rank-one-fidelity-lipschitz}\\
|\Deltaopt(\rho,\NN)-\Deltaopt(\rho,\NN')|
&\leq2\dnorm{\NN-\NN'}.
\label{eq:optimized-gap-lipschitz}
\end{align}
\end{subequations}

If $\NN$ and $\NN'$ are CPTP maps, let $D_k:\HC\to\HL$ be Kraus
operators of $\DD$. On a failure space $\HH_F\cong\HC$, set
\begin{equation}\label{eq:sm-decoder-completion}
F_D:=\left(I_C-\sum_kD_k^{\dagger}D_k\right)^{1/2},
\qquad
\DD_{\mathrm{ex}}(X):=\DD(X)\oplus F_DXF_D^{\dagger}.
\end{equation}
Through a fixed identification $\HC\cong\HH_F$, the second block maps into the
orthogonal failure space, and
$\DD_{\mathrm{ex}}:\Lin(\HC)\to\Lin(\HL\oplus\HH_F)$. Both blocks are completely
positive~(CP), and $\sum_kD_k^{\dagger}D_k+F_D^{\dagger}F_D=I_C$, so
$\DD_{\mathrm{ex}}$ is CPTP, with $\DD(X)$ as its success block.

Embed the purification projector in the success block as
$\hat{\Pi}_{\mathrm{pass}}:=\kb{\psi_\rho}\oplus0_{\HRef\otimes\HH_F}$ and set
$S:=J_{\rho}(\DD_{\mathrm{ex}}\circ(\NN-\NN')\circ\CCal)$. Because
$\NN-\NN'$ is trace annihilating, $S$ is Hermitian and traceless. The
variational trace distance bound therefore gives
\begin{equation}\label{eq:fixed-lipschitz-cptp}
|f_{\CCal}(\NN)-f_{\CCal}(\NN')|
=|\tr(\hat{\Pi}_{\mathrm{pass}}S)|\leq\frac{1}{2}\|S\|_1
\leq\frac{1}{2}\dnorm{\DD_{\mathrm{ex}}\circ(\NN-\NN')\circ\CCal}
\leq\frac{1}{2}\dnorm{\NN-\NN'}.
\end{equation}
Here the equality uses the zero failure support of $\hat{\Pi}_{\mathrm{pass}}$,
and the inequalities follow from the variational trace-distance bound for a
traceless Hermitian operator and an effect, the diamond-norm definition, and
submultiplicativity. Taking the same suprema gives bounds with coefficient one
half corresponding to Eqs.~\eqref{eq:unrestricted-fidelity-lipschitz}
and~\eqref{eq:rank-one-fidelity-lipschitz}. The triangle inequality then gives
the bound with coefficient one in Eq.~\eqref{eq:optimized-gap-lipschitz}, as
stated in the theorem.
\end{proof}

\section{The symmetric high-rank construction}\label{sec:sm-family}

To restate the construction for general logical dimension, let
$d_S=d_L\geq2$, $d_C=d_L^2$, $\HC=\HL\otimes\HS$, and $\tau:=I_L/d_L$.
After identifying $\HRef\otimes\HL$ with $\HC$, the canonical purification
$\ket{\psi_\tau}:=\dket{\sqrt\tau}\in\HC$ is maximally entangled. Set
$\hat{\pi}_C=\kb{\psi_\tau}_{L\otimes S}$. For
$0\leq p<1$, the five coefficients are
\begin{equation}\label{eq:sm-family-lambdas}
\begin{aligned}
\lambda_1&:=1-\frac{1-d_L^2\lambda_5}{d_L^2-2},&
\lambda_2&:=\frac{(1-\lambda_1)(d_L^2\lambda_1-1)}
{d_L^2\lambda_4},&
\lambda_3&:=\frac{2(1-\lambda_1)^2}{d_L^2\lambda_4},\\
\lambda_4&:=\frac{d_L^2-1+p}{d_L^2},&
\lambda_5&:=\frac{1-p}{d_L^2}.
\end{aligned}
\end{equation}
The five coefficients satisfy
$\lambda_1\geq\lambda_2\geq\lambda_3\geq0$,
$\lambda_4>\lambda_5>0$,
$\lambda_1+\lambda_2+(d_L^2-1)\lambda_3=1$, and
$\lambda_4+\lambda_5=1$.

Let $Q_p$ be given by
\begin{equation}\label{eq:sm-family-generator}
Q_p:=
\sqrt{d_L}\left[
\sqrt{\frac{\lambda_4}{d_L^2-1}}(I_C-\hat{\pi}_C)
+\sqrt{\lambda_5}\hat{\pi}_C\right].
\end{equation}
Its eigenvalues are
$\sqrt{d_L\lambda_4/(d_L^2-1)}$ and $\sqrt{d_L\lambda_5}$, so
$Q_p\succ0$ throughout this interval.

The encoder and decoder have Kraus operators
\begin{equation}\label{eq:sm-family-pair}
C_i(p):=Q_p(I_L\otimes\ket{i}_S),
\qquad
D_j:=I_L\otimes\bra{j}_S,
\end{equation}
for $i,j=0,\ldots,d_L-1$. Then $\CCal_p^{\opt}$ and $\DD=\trS$ are quantum
channels, where $\CCal_p^{\opt}(\omega)=Q_p(\omega\otimes I_S)Q_p$ for every
$\omega\in\Lin(\HL)$. Let $B_{ij}(p):=C_i(p)D_j$, let $P_p$ be the
orthogonal projector onto $\operatorname{span}\{\dket{B_{ij}(p)}\}$, and write
$P_p^{\perp}=I-P_p$. The vectors
$\{\dket{B_{ij}(p)}\}_{i,j=0}^{d_L-1}$ form an orthonormal basis of
$\Ran{P_p}$, and
$\tr_1P_p=d_L Q_p^2$, where $\tr_1$ denotes the partial trace over the first
subsystem of $\HC\otimes\HC$.

Let $X_p$ be given by
\begin{equation}\label{eq:sm-xp-closed}
X_p=\lambda_3 I_C+
\frac{\lambda_2}{\lambda_4}\hat{\pi}_C.
\end{equation}
It satisfies
\begin{equation}\label{eq:sm-xp-system}
\tr_1\left[P_p^{\perp}(I_C\otimes X_p)P_p^{\perp}\right]
=I_C-d_L\lambda_1Q_p^2.
\end{equation}
The noise map and its certificate operator are given by
\begin{align}
J(\NN_p^{\dagger})
&:=\lambda_1P_p+P_p^{\perp}(I_C\otimes X_p)P_p^{\perp},
\label{eq:sm-family-noise}\\
Z_p&:=\lambda_1I_{C\otimes C}-J(\NN_p^{\dagger})
=P_p^{\perp}(I_C\otimes M_p)P_p^{\perp},
\label{eq:sm-family-z}\\
M_p&:=\lambda_1I_C-X_p.
\label{eq:sm-family-weight}
\end{align}
Together with $\tr_1P_p=d_L Q_p^2$, Eq.~\eqref{eq:sm-xp-system} gives
$\tr_1J(\NN_p^{\dagger})=I_C$. Under the input-first Choi matrix convention, this
means that $\NN_p^{\dagger}$ is unital, equivalently that $\NN_p$ is TP.
For $d_L\geq2$ and $0<p<1$, Eqs.~\eqref{eq:sm-xp-closed}
and~\eqref{eq:sm-family-weight} show that the smallest eigenvalue of $M_p$ is
$m_p:=\lambda_1-\lambda_3-\lambda_2/\lambda_4$, which is positive by direct
substitution, so $M_p\succeq m_pI_C\succ0$. From
Eq.~\eqref{eq:sm-xp-closed}, we also obtain $X_p\succ0$ and hence
$J(\NN_p^{\dagger})\succeq0$, while
Eq.~\eqref{eq:sm-family-z} gives $Z_p\succeq0$ and
\begin{equation}\label{eq:sm-z-kernel}
\ker Z_p=\Ran{P_p}
=\operatorname{span}\{\dket{B_{ij}(p)}\}.
\end{equation}
At $p=0$, the formulas give $Q_0=I_C/\sqrt{d_L}$, $X_0=0$,
$J(\NN_0^\dagger)=P_0$, and $Z_0=P_0^\perp$, so positivity and the kernel
characterization persist, and the resulting channel
$\NN_0(\omega)=\trS(\omega)\otimes I_S/d_L$ is perfectly recoverable with
$V_0=I_L\otimes\ket0$.
Thus $\NN_p$ is CPTP throughout $0\leq p<1$.

\subsection{Optimality certificate}
\label{sec:sm-certificate}

For the maximally mixed input $\tau$, define the channel fidelity by
$\Fch(\Phi):=\Fe(\tau,\Phi)$. For any CPTNI encoder--decoder pair, this
fidelity depends only on the CPTNI map $\BCal=\CCal\circ\DD$ on $\HC$. Since
$J(\BCal)\succeq0$ and $\tr J(\BCal)\leq d_L^2$,
\begin{equation}\label{eq:sm-certificate-bound}
\Fch(\DD\circ\NN_p\circ\CCal)=\frac{\lambda_1}{d_L^2}\tr J(\BCal)
-\frac{1}{d_L^2}\tr[Z_pJ(\BCal)]
\leq\lambda_1.
\end{equation}
For the pair in Eq.~\eqref{eq:sm-family-pair}, $\BCal_p$ is TP and every
$\dket{B_{ij}(p)}$ lies in $\ker Z_p$, so the bound is saturated. Hence
\begin{equation}\label{eq:sm-family-optimum}
\Fopt(\tau,\NN_p)=\lambda_1,
\qquad
\epsopt(\tau,\NN_p)=1-\lambda_1=\frac{p}{d_L^2-2}.
\end{equation}

Equality in Eq.~\eqref{eq:sm-certificate-bound} forces
$\tr J(\BCal)=d_L^2$ and $\tr[Z_pJ(\BCal)]=0$. The first equality makes the
CPTNI map $\BCal$ TP, while positivity of $Z_p$ and $J(\BCal)$ makes the
second imply $\supp(J(\BCal))\subseteq\ker Z_p$. Any Kraus representation
$\{\widetilde{B}_\ell\}_\ell$ of $\BCal$ satisfies
$\operatorname{span}\{\dket{\widetilde{B}_\ell}\}_\ell
=\supp(J(\BCal))$, so Eq.~\eqref{eq:sm-z-kernel} gives
$\dket{\widetilde{B}_\ell}\in
\operatorname{span}\{\dket{B_{ij}(p)}\}_{i,j}$ for every $\ell$. We may thus
write
$\widetilde{B}_\ell=\sum_{i,j}h_{ij}^{(\ell)}B_{ij}(p)$ and set
$H_{ii'}^{jj'}:=\sum_\ell(h_{ij}^{(\ell)})^*h_{i'j'}^{(\ell)}$. Using
$B_{ij}(p)^{\dagger}B_{i'j'}(p)
=(Q_p^2)_{ii'}\otimes\ket{j}_S\!\bra{j'}_S$,
matching the coefficients of the linearly independent matrix units
$\ket{j}_S\!\bra{j'}_S$ in the TP condition
$\sum_\ell\widetilde{B}_\ell^{\dagger}\widetilde{B}_\ell=I_C$ gives
\begin{equation}\label{eq:sm-trace-system}
\sum_{i,i'}H_{ii'}^{jj'}(Q_p^2)_{ii'}=\delta_{jj'}I_L,
\end{equation}
where
\begin{equation}\label{eq:sm-qp-blocks}
(Q_p^2)_{ij}
=\frac{d_L\lambda_4}{d_L^2-1}\delta_{ij}I_L
-\frac{p}{d_L^2-1}\ket i\!\bra j,
\qquad
(Q_p^2)_{ij}:\HL\to\HL,
\quad 0\leq i,j<d_L.
\end{equation}
For $p\in(0,1)$, every off-diagonal block is a nonzero multiple of a distinct
matrix unit. The diagonal entries then force all coefficients of the diagonal blocks
in a vanishing linear combination to be zero. Hence these $d_L^2$ blocks are
linearly independent. Since $C_i(p)^{\dagger}C_j(p)=(Q_p^2)_{ij}$, Choi's
extremality criterion for quantum channels~\cite{Choi1975} shows that
$\CCal_p^{\opt}$ is an extreme point of the convex set of CPTP maps for $p\in(0,1)$.
The linear independence also makes Eq.~\eqref{eq:sm-trace-system} have a
unique solution for $H$. In the
orthonormal basis $\{\dket{B_{ij}(p)}\}_{i,j}$ of
$\Ran{P_p}$, the matrix entries of $J(\BCal)$ are
$(H_{ii'}^{jj'})^*$. Since
$\BCal_p:=\CCal_p^{\opt}\circ\DD^{\opt}$ supplies a solution,
uniqueness gives
$J(\BCal)=J(\BCal_p)$ and hence $\BCal=\BCal_p$. Thus every
optimal pair has composite $\BCal_p$ for $p\in(0,1)$.

\subsection{General form of optimal encoder--decoder pairs}
\label{sec:sm-optimal-pairs}

For $p\in(0,1)$, the preceding subsection shows that every optimal pair has
composite $\BCal_p=\CCal_p^{\opt}\circ\DD^{\opt}$. Consider an arbitrary CPTNI
factorization $\CCal\circ\DD'=\CCal_p^{\opt}\circ\trS$. We show that its only
freedom is a unitary channel on $\HL$. Since $\CCal$ and $\DD'$ are CPTNI
while their composite is TP,
\begin{equation}
I_C=\DD'^{\dagger}(\CCal^{\dagger}(I_C))
\preceq\DD'^{\dagger}(I_L)\preceq I_C.
\end{equation}
Both inequalities are saturated. In particular,
$\DD'^{\dagger}(I_L)=I_C$, so $\DD'$ is TP. If $\DD'(X)=0$, injectivity of
$\CCal_p^{\opt}$, which follows from $Q_p\succ0$, gives $\trS X=0$ and hence
$\ker\DD'\subseteq\ker\trS$. Surjectivity of $\trS$ and the rank--nullity
theorem applied to $\DD'$ give
$d_L^4-d_L^2=\dim\ker\trS\geq\dim\ker\DD'\geq d_L^4-d_L^2$.
Thus the kernels agree and $\DD'$ is surjective.

Let $\KK(\omega):=\omega\otimes\kb0$, which is a quantum channel satisfying
$\trS\circ\KK=\id_L$, and set
$\Lambda:=\DD'\circ\KK$. For every $X\in\Lin(\HC)$,
$X-\KK(\trS X)$ lies in the common kernel, so
$\DD'(X)=\Lambda(\trS X)$. Substituting this identity into the original
factorization and using surjectivity of $\trS$ gives
\begin{equation}
\DD'=\Lambda\circ\trS,
\qquad \CCal\circ\Lambda=\CCal_p^{\opt}.
\end{equation}
The first identity and surjectivity of $\DD'$ make $\Lambda$ surjective, so it
is invertible on the finite-dimensional space $\Lin(\HL)$. Its definition as
$\DD'\circ\KK$ also shows that it is CPTP.

Since $Q_p\succ0$, $\CCal_p^{\opt}$ admits the CP left inverse
\begin{equation}
\bigl(\CCal_p^{\opt}\bigr)^{-1}(Y):=\frac{1}{d_L}
\trS\left(Q_p^{-1}YQ_p^{-1}\right),
\qquad Y\in\Lin(\HC).
\end{equation}
Here $\bigl(\CCal_p^{\opt}\bigr)^{-1}$ denotes a CP extension to $\Lin(\HC)$
of the inverse of $\CCal_p^{\opt}$ on its range, and direct substitution gives
$\bigl(\CCal_p^{\opt}\bigr)^{-1}\circ\CCal_p^{\opt}=\id_L$. Applying it to
$\CCal\circ\Lambda=\CCal_p^{\opt}$ gives
$\Lambda^{-1}=\bigl(\CCal_p^{\opt}\bigr)^{-1}\circ\CCal$, which is CP. Since
$\Lambda$ is TP and
bijective, its inverse is also TP, and hence $\Lambda$ is a unitary channel.
Thus the optimal pairs for $p\in(0,1)$ are exactly
\begin{equation}\label{eq:sm-optimal-gauge}
(\CCal,\DD')=(\CCal_p^{\opt}\circ\UCal^{\dagger},
\UCal\circ\trS),
\end{equation}
where $\UCal$ is a unitary channel. Conversely, every such pair has
composite $\BCal_p$ and is optimal.
Since $Q_p\succ0$, every optimal encoder in
Eq.~\eqref{eq:sm-optimal-gauge} maps a pure logical state $\rho$ to
$Q_p[\UCal^{\dagger}(\rho)\otimes I_S]Q_p$, which has rank $d_L\geq2$ and is
therefore a mixed code state.

\subsection{Unitality and covariance of the constructed quantum channel}
\label{sec:sm-noise-symmetry}

We now prove that $\NN_p$ is unital and covariant under the conjugate tensor
representation for every $0\leq p<1$. Fix $u\in U(d_L)$ and set
$W(u):=u\otimes u^*$ on $\HC$, where complex conjugation is taken in the
basis paired by $\ket{\psi_\tau}$. We claim that
\begin{equation}\label{eq:sm-noise-covariance}
\NN_p(W(u)\omega W(u)^{\dagger})
=W(u)\NN_p(\omega)W(u)^{\dagger}
\end{equation}
for every $\omega\in\Lin(\HC)$. Invariance of the maximally entangled state
under $W(u)$ gives
$[W(u),\hat{\pi}_C]=[W(u),Q_p]=[W(u),X_p]=0$. Moreover,
under input-first vectorization, conjugation by $W(u)$ is represented by
$W^*(u)\otimes W(u)$, with
$[W^*(u)\otimes W(u)]\dket{M}=\dket{W(u)MW(u)^{\dagger}}$. Hence
\begin{equation}\label{eq:sm-frame-covariance}
[W^*(u)\otimes W(u)]\dket{B_{ij}(p)}
=\sum_{k,\ell=0}^{d_L-1}u_{ki}^*u_{\ell j}\dket{B_{k\ell}(p)}.
\end{equation}
This unitary transformation law leaves $\Ran{P_p}$ invariant, giving
$[W^*(u)\otimes W(u),P_p]=0$ and hence the same commutation with $P_p^\perp$.
Moreover, $[W(u),X_p]=0$ implies
$[W^*(u)\otimes W(u),I_C\otimes X_p]=0$.
Using these commutation relations in
Eq.~\eqref{eq:sm-family-noise} gives
$[J(\NN_p^{\dagger}),W^*(u)\otimes W(u)]=0$, so the covariance criterion for
Choi matrices proves covariance of $\NN_p^{\dagger}$. Taking adjoints and
using that
$\{W(u):u\in U(d_L)\}$ is closed under adjoints proves
Eq.~\eqref{eq:sm-noise-covariance}.

It remains to prove unitality. Let $\tr_2$ denote the trace over the second
factor of the Choi matrix. For operators $A$ and $B$, input-first
vectorization yields
$\tr_2[\dket{A}\dbra{B}]=(A^{\dagger}B)^*$ and hence
$\tr_2J(\NN_p^{\dagger})=[\NN_p(I_C)]^*$. Since the $B_{ij}(p)$ are
Kraus operators of the quantum channel
$\BCal_p=\CCal_p^{\opt}\circ\DD^{\opt}$, we have
\begin{equation}\label{eq:sm-second-marginal-projector}
\tr_2P_p
=\left[\sum_{i,j}B_{ij}(p)^{\dagger}B_{ij}(p)\right]^*=I_C,
\qquad
\tr_2P_p^{\perp}=(d_L^2-1)I_C.
\end{equation}

Set $\hat{\eta}:=I_C\otimes\hat{\pi}_C$ and
$\hat{\psi}_{ij}:=\ket{\psi_\tau}\!\bra{ij}_C$. The vectors
$\{\dket{\hat{\psi}_{ij}}\}_{i,j}$ form an
orthonormal basis for the range of $\hat{\eta}$. Substituting the expression
for $Q_p$ from Eq.~\eqref{eq:sm-family-generator} yields
\begin{equation}\label{eq:sm-frame-overlap}
\begin{aligned}
\langle\!\langle B_{k\ell}(p)\vert\hat{\psi}_{ij}\rangle\!\rangle
&=\tr[B_{k\ell}(p)^{\dagger}\hat{\psi}_{ij}]
=\sqrt{\lambda_5}\delta_{ik}\delta_{j\ell},\\
B_{ij}(p)^{\dagger}\hat{\psi}_{ij}
&=\sqrt{\lambda_5}\kb{ij}_C.
\end{aligned}
\end{equation}
Together with the orthonormal basis
$\{\dket{B_{ij}(p)}\}_{i,j}$ of $\Ran{P_p}$, these identities imply
\begin{equation}\label{eq:sm-projected-eta-identities}
P_p\hat{\eta}P_p=\lambda_5P_p,
\qquad
\tr_2(P_p\hat{\eta})=\tr_2(\hat{\eta}P_p)=\lambda_5I_C.
\end{equation}
Expanding $P_p^{\perp}=I-P_p$ and using
Eqs.~\eqref{eq:sm-second-marginal-projector}
and~\eqref{eq:sm-projected-eta-identities}, together with
$\hat{\eta}=I_C\otimes\hat{\pi}_C$ and $\tr\hat{\pi}_C=1$, we obtain
\begin{equation}\label{eq:sm-complementary-marginal}
\tr_2(P_p^{\perp}\hat{\eta}P_p^{\perp})
=(1-\lambda_5)I_C
=\lambda_4I_C.
\end{equation}

Using Eqs.~\eqref{eq:sm-family-noise},~\eqref{eq:sm-xp-closed},~\eqref{eq:sm-second-marginal-projector},
and~\eqref{eq:sm-complementary-marginal}, the second marginal is
\begin{equation}\label{eq:sm-unital-derivation}
\begin{aligned}
\tr_2J(\NN_p^{\dagger})
&=\lambda_1\tr_2P_p+\lambda_3\tr_2P_p^{\perp}
+\frac{\lambda_2}{\lambda_4}
\tr_2(P_p^{\perp}\hat{\eta}P_p^{\perp})\\
&=\left[\lambda_1+(d_L^2-1)\lambda_3
+\frac{\lambda_2}{\lambda_4}\lambda_4\right]I_C\\
&=[\lambda_1+\lambda_2+(d_L^2-1)\lambda_3]I_C=I_C.
\end{aligned}
\end{equation}
The input-first identity above gives $\NN_p(I_C)=I_C$. Together with complete
positivity and the TP property established above, this shows that $\NN_p$ is a
unital and covariant CPTP map.

\subsection{The exact quadratic gap coefficient}\label{sec:sm-kappa}

The goal is to determine the globally optimized quadratic coefficient as
$p\to0^+$. We cast the rank-one gap as a compact residual problem and evaluate
its global tangent minimum. This route covers arbitrary encoder--decoder pairs
that depend on $p$, including changes in the decoder Choi matrix rank.

\medskip
\noindent\emph{Compact residual formulation.}
When $d_C\geq d_L$, a rank-one CPTNI encoder has the form
$\CCal_V=V(\cdot)V^{\dagger}$ with $V^{\dagger}V\preceq I_L$. For a state
$\rho$ and a fixed decoder, write
\begin{equation}\label{eq:sm-rank-one-objective}
K_{\CCal}:=J_{\rho^2}((\DD\circ\NN)^{\dagger})
\in\Pos(\HRef\otimes\HC),
\qquad
\widetilde{F}(V):=\Fe(\rho,\DD\circ\NN\circ\CCal_V)
=\dbra{V}K_{\CCal}\dket{V}.
\end{equation}

\begin{lemma}[partial-isometry reduction]
\label{lem:sm-partial-isometry-reduction}
Let $\rho$ be a state on $\HL$ with support projector $\PiRho$, and assume
$d_C\geq d_L$. If $K\in\Pos(\HRef\otimes\HC)$ and
$q(V):=\dbra{V}K\dket{V}$ satisfies $q(V)=q(V\PiRho)$, then the maximum of
$q$ over contractions $V:\HL\to\HC$ is attained by a partial isometry with
initial space $\Hr$. Equivalently, $V=V\PiRho$ and
$V^{\dagger}V=\PiRho$.
\end{lemma}

\begin{proof}
For any contraction $V$, set $V_\rho:=V\PiRho$ and
$S_\rho:=(V_\rho^{\dagger}V_\rho)^{1/2}$. Extend the polar partial isometry of
$V_\rho$ to a partial isometry $U_\rho:\HL\to\HC$ satisfying
$V_\rho=U_\rho S_\rho$ and $U_\rho^{\dagger}U_\rho=\PiRho$. This extension
exists because $d_C\geq d_L\geq\dim\Hr$. Since
$0\preceq S_\rho\preceq\PiRho$, let
$V_{\pm}:=U_\rho[S_\rho\pm i(\PiRho-S_\rho^2)^{1/2}]$. These operators obey
$V_{\pm}^{\dagger}V_{\pm}=\PiRho$ and
$V_\rho=(V_++V_-)/2$. Thus they are partial isometries with initial space
$\Hr$. Moreover, positivity of $K$ gives
$[q(V_+)+q(V_-)]/2-q(V_\rho)=q(V_+-V_-)/4\geq0$. Since
$q(V)=q(V_\rho)$, at least one of $V_\pm$ has objective value no smaller than
that of $V$. A maximizer exists by compactness of the contraction set, and
applying this replacement to it proves the claim.
\end{proof}

For $K_{\CCal}$ in Eq.~\eqref{eq:sm-rank-one-objective}, the support identity
$K_{\CCal}=(\PiRho^{\Tp}\otimes I_C)K_{\CCal}
(\PiRho^{\Tp}\otimes I_C)$ verifies the lemma's hypothesis.

For the remainder of this subsection, we specialize to the family
$(\tau,\NN_p)$ defined in the main text, with $d_L\geq2$ and $d_C=d_L^2$.
Because $\tau$ has full rank, an optimal rank-one encoder may therefore be
chosen as a partial-isometry encoder $\CCal_V$ with $V^{\dagger}V=I_L$.
Completing a CPTNI decoder $\DDtilde$ to a TP map cannot decrease the fidelity.
Indeed, adding the CP map
$X\mapsto\tr([I_C-\DDtilde^{\dagger}(I_L)]X)\kb{0}_L$ makes the decoder
TP and adds a nonnegative term to the objective. Thus an optimal decoder may
also be taken TP.
Set $d_E:=d_C d_L=d_L^3$, the maximal decoder Choi matrix rank, and
$\HE:=\Cbb^{d_E}$ with fixed orthonormal basis
$\{\ket{m}_E\}_{m=0}^{d_E-1}$.
Every CPTP decoder $\DDtilde:\Lin(\HC)\to\Lin(\HL)$ has a Stinespring
dilation $\Dtilde:\HC\to\HL\otimes\HE$ satisfying
$\Dtilde^{\dagger}\Dtilde=I_C$, with Kraus operators
$\bigl[\Dtilde\bigr]_m:=(I_L\otimes\bra{m}_E)\Dtilde$. Zero-padding shorter
Kraus representations includes every decoder Choi matrix rank in this fixed,
$p$-independent parameter space.

The compact feasible set is the product of two \emph{Stiefel manifolds}
\begin{equation}\label{eq:sm-stiefel-product}
\StiefelProd:=\{x=(\Dtilde,V):
\Dtilde^{\dagger}\Dtilde=I_C,V^{\dagger}V=I_L\}.
\end{equation}
Associate with each $x=(\Dtilde,V)\in\StiefelProd$ the residual vector
\begin{equation}\label{eq:sm-residual}
\ket{r(p,x)}
:=\sum_{m=0}^{d_E-1}\ket{m}_E\otimes
(I_C\otimes M_p^{1/2})P_p^{\perp}
\dket{V\bigl[\Dtilde\bigr]_m}.
\end{equation}
For the tangent analysis, regard the auxiliary residual space
$\HH_{E\otimes C\otimes C}:=\HE\otimes\HC\otimes\HC$ as a real inner-product
space by taking the real part, which leaves its norm unchanged. The auxiliary
certificate
channel $\BCal:=\CCal_V\circ\DDtilde$ is TP, has Kraus operators
$V\bigl[\Dtilde\bigr]_m$, and does not describe the physical order of the
encoder, noise, and decoder.
Consequently, the fixed-pair fidelity loss and the optimized gap satisfy
\begin{subequations}\label{eq:sm-optimized-gap-leakage}
\begin{align}
\Fopt(\tau,\NN_p)-\Fch(\DDtilde\circ\NN_p\circ\CCal_V)
&=\frac{1}{d_L^2}\tr[Z_pJ(\BCal)]
=\frac{1}{d_L^2}\lVert\ket{r(p,x)}\rVert_2^2,
\label{eq:sm-optimized-gap-leakage-a}\\
\Deltaopt(\tau,\NN_p)
&=\frac{1}{d_L^2}\min_{x\in\StiefelProd}
\lVert\ket{r(p,x)}\rVert_2^2.
\label{eq:sm-optimized-gap-leakage-b}
\end{align}
\end{subequations}
Because $\BCal$ is TP, Eqs.~\eqref{eq:sm-certificate-bound}
and~\eqref{eq:sm-family-optimum}, together with
$J(\BCal)=\sum_m\dket{V[\Dtilde]_m}\dbra{V[\Dtilde]_m}$ and
$Z_p=P_p^{\perp}(I_C\otimes M_p)P_p^{\perp}$, give the first line. Minimizing
it gives the second. The explicit formulas for $Q_p$, $P_p$, and $M_p$ extend
smoothly to $|p|<p_0$. Since $M_0=I_C$, continuity gives $M_p\succ0$ there, so
$p\mapsto M_p^{1/2}$ and $(p,x)\mapsto\ket{r(p,x)}$ are smooth near $p=0$.
This two-sided extension is used only for Taylor expansions, while the physical
limit is $p\to0^+$.

\medskip
\noindent\emph{Perfect pairs and symmetry reduction.}
Let
$\StiefelProd_0:=\{x\in\StiefelProd:\ket{r(0,x)}=0\}$. We first characterize
this set and then use symmetry to reduce all subsequent calculations to one
canonical pair.
At $p=0$, $M_0=I_C$ and
$\Ran{P_0}=\{\dket{I_L\otimes A}:A\in\Lin(\HS)\}$. Hence
$\ket{r(0,x)}=0$ exactly when there are operators $A_m\in\Lin(\HS)$ such that
$V\bigl[\Dtilde\bigr]_m=I_L\otimes A_m$ for every $m$.
Trace preservation makes
$\ECal(X):=\sum_m A_m X A_m^{\dagger}$ a quantum channel and
gives $\BCal=\id_L\otimes\ECal$. The operator
$\BCal(I_C)=I_L\otimes\ECal(I_S)$ is supported on the
space $\Ran{V}$ of dimension $d_L$, so $\ECal(I_S)$ has rank at most one.
Since $\tr[\ECal(I_S)]=d_L$, it has rank one. Thus the range of every Kraus
operator $A_m$ is contained in one common subspace of dimension one, so
$A_m=\ket{\chi}\!\bra{a_m}$ for a normalized $\ket{\chi}\in\HS$. Trace
preservation gives $\sum_m\kb{a_m}=I_S$ and hence
$\ECal(X)=\tr(X)\kb{\chi}$.

The support of $\BCal(I_C)$ is therefore
$\HL\otimes\operatorname{span}\{\ket{\chi}\}$, and it has the same
dimension as $\Ran{V}$. Equality of the supports gives
$V=(I_L\otimes\ket{\chi})U_L$ for a unitary $U_L$ on $\HL$. Multiplying
$\BCal(X)=V\DDtilde(X)V^{\dagger}$ by $V^{\dagger}$ and $V$ then gives
$\DDtilde(X)=U_L^{\dagger}\trS(X)U_L$. Let
$D_{\trS}:=\sum_{j=0}^{d_L-1}D_j\otimes\ket{j}_S$ be the minimal
Stinespring dilation of $\trS$. Since $D_{\trS}$ is minimal, the
Stinespring uniqueness theorem gives
$\Dtilde=(U_L^{\dagger}\otimes U_{S\to E})D_{\trS}$ for an isometry
$U_{S\to E}:\HS\to\HE$ satisfying
$U_{S\to E}^{\dagger}U_{S\to E}=I_S$. Consequently,
every $x=(\Dtilde,V)\in\StiefelProd_0$ can be represented by a normalized
$\ket{\chi}\in\HS$, a unitary $U_L\in U(d_L)$, and such an isometry
$U_{S\to E}$
as
\begin{equation}\label{eq:sm-zero-set}
\StiefelProd_0=\bigl\{\bigl((U_L^{\dagger}\otimes U_{S\to E})D_{\trS},
(I_L\otimes\ket{\chi})U_L\bigr):
\langle\chi\vert\chi\rangle=1,
U_L\in U(d_L),U_{S\to E}^{\dagger}U_{S\to E}=I_S\bigr\}.
\end{equation}
Conversely, every triple satisfying the displayed conditions gives a pair
that attains the unrestricted optimum at $p=0$.

For $\boldsymbol g:=(g_1,g_2,g_E)\in
U(d_L)\times U(d_L)\times U(d_E)$ and $x=(\Dtilde,V)$, define
\begin{subequations}\label{eq:sm-symmetry-transformations}
\begin{align}
\boldsymbol g(\Dtilde,V)
&:=\bigl((g_2^{\dagger}\otimes g_E)\Dtilde W(g_1)^{\dagger},W(g_1)Vg_2\bigr),
\label{eq:sm-symmetry-action}\\
(\ket{\chi},U_L,U_{S\to E})
&\longmapsto
(g_1^*\ket{\chi},g_1U_Lg_2,g_EU_{S\to E}g_1^{\Tp}),
\label{eq:sm-symmetry-parameters}\\
\ket{r(p,\boldsymbol g(x))}
&=(g_E\otimes W^*(g_1)\otimes W(g_1))\ket{r(p,x)}.
\label{eq:sm-residual-equivariance}
\end{align}
\end{subequations}
The unitary factors preserve the Stiefel constraints and Frobenius distances,
so this is a Frobenius-isometric action on all of $\StiefelProd$. The parameter
transformation preserves Eq.~\eqref{eq:sm-zero-set}, so the action maps
$\StiefelProd_0$ onto itself, and the residual equivariance holds throughout
$\StiefelProd$. Since its multiplier is unitary,
$\lVert\ket{r(p,\boldsymbol g(x))}\rVert_2
=\lVert\ket{r(p,x)}\rVert_2$ for every $p$.

Let $\hat{\mu}:\HS\hookrightarrow\HE$ be the canonical embedding specified by
$\hat{\mu}\ket{j}_S=\ket{j}_E$. For any point of $\StiefelProd_0$, choose
$g_1^*\ket{\chi}=\ket{0}$, $g_2=U_L^{\dagger}g_1^{\dagger}$, and
$g_EU_{S\to E}g_1^{\Tp}=\hat{\mu}$. The last choice exists because a unitary
on $\HE$ can map one orthonormal $d_L$ frame to another. Thus the action is
transitive on $\StiefelProd_0$ and maps every point to the canonical pair
\begin{equation}\label{eq:sm-canonical-zero}
x_0:=(\Dtilde_0,V_0),\qquad
\left\{
\begin{aligned}
\Dtilde_0&:=(I_L\otimes\hat{\mu})D_{\trS},\\
V_0&:=I_L\otimes\ket{0}_S.
\end{aligned}
\right.
\end{equation}
Residual equivariance identifies the first-order calculations at all points of
$\StiefelProd_0$, so the kernel and tangent minimum may be evaluated at $x_0$.
This zero set is a compact group orbit and hence an embedded submanifold of
$\StiefelProd$. We denote the tangent space of $\StiefelProd$ at $x_0$ by
$\TCal_0:=T_{x_0}\StiefelProd$.

\begin{lemma}[tangent vectors with vanishing quadratic coefficient]
\label{lem:sm-tangent-characterization}
Let $\LMap_0:\TCal_0\to\HH_{E\otimes C\otimes C}$ be the real-linear
first-order residual map at $x_0$. For any smooth feasible curve
$x(t)\in\StiefelProd$ with $x(0)=x_0$, its action on
$\dot x(0)\in\TCal_0$ is
$\LMap_0(\dot x(0)):=\left.\frac{\mathrm d}{\mathrm dt}
\ket{r(0,x(t))}\right|_{t=0}$. Its kernel satisfies
\begin{equation}\label{eq:sm-kernel-tangent-space}
\ker(\LMap_0)=T_{x_0}\StiefelProd_0.
\end{equation}
Moreover,
$\lVert\ket{r(0,x(t))}\rVert_2^2
=t^2\lVert\LMap_0(\dot x(0))\rVert_2^2+o(t^2)$, and the $p=0$ fidelity difference is
$\frac{t^2}{d_L^2}\lVert\LMap_0(\dot x(0))\rVert_2^2+o(t^2)$. Hence its quadratic
coefficient vanishes exactly for
$\dot x(0)\in T_{x_0}\StiefelProd_0$ and is strictly positive otherwise.
\end{lemma}

\begin{proof}
Let $x(t)=(\Dtilde(t),V(t))$ be a smooth feasible curve with $x(0)=x_0$.
At $t=0$, the dilation blocks satisfy
$\bigl[\Dtilde_0\bigr]_j=D_j$ for $0\leq j<d_L$, whereas
$\bigl[\Dtilde_0\bigr]_m=0$ for $m\geq d_L$.
We call $0\leq m<d_L$ and $d_L\leq m<d_E$ the active and inactive environment
indices, respectively. Dots denote derivatives at $t=0$, so
$\dot{\Dtilde}:=\left.\frac{\mathrm d}{\mathrm dt}\Dtilde(t)\right|_{t=0}$
and
$\bigl[\dot{\Dtilde}\bigr]_m:=(I_L\otimes\bra{m}_E)\dot{\Dtilde}$.
For block extraction, write
$\bigl[\dot{\Dtilde}\bigr]_{m,k}:=
(I_L\otimes\bra{m}_E)\dot{\Dtilde}(I_L\otimes\ket{k}_S)$ and
$\bigl[\dot V\bigr]_i:=(I_L\otimes\bra{i}_S)\dot V$.
The blocks $\bigl[\dot V\bigr]_i$ and
$\bigl[\dot{\Dtilde}\bigr]_{m,k}$ lie in
$\Lin(\HL)$. Every
$(\dot{\Dtilde},\dot V)\in\TCal_0$ satisfies
$\Dtilde_0^{\dagger}\dot{\Dtilde}
+\dot{\Dtilde}^{\dagger}\Dtilde_0=0$ and
$V_0^{\dagger}\dot V+\dot V^{\dagger}V_0=0$. In blocks, these conditions are
\begin{equation}\label{eq:sm-tangent-constraints}
\bigl[\dot{\Dtilde}\bigr]_{j,k}
+\left(\bigl[\dot{\Dtilde}\bigr]_{k,j}\right)^{\dagger}=0
\quad(0\leq j,k<d_L),
\qquad
\bigl[\dot V\bigr]_0+\left(\bigl[\dot V\bigr]_0\right)^{\dagger}=0,
\end{equation}
with no constraint on $\bigl[\dot{\Dtilde}\bigr]_{m,k}$ for $m\geq d_L$.
Conversely, these conditions characterize the tangent space of the two
Stiefel factors, so every such pair belongs to $\TCal_0$ and is realized by a
smooth feasible curve.

For $X\in\Lin(\HL)$, its traceless part is
\begin{equation}\label{eq:sm-traceless-part}
\tracelesspart{X}:=X-\frac{\tr(X)}{d_L}I_L.
\end{equation}
For $Y\in\Lin(\HC)$ and
$0\leq i,k<d_L$, write the blocks and their norm identity as follows:
\begin{subequations}\label{eq:sm-normal-blocks}
\begin{align}
\bigl[Y\bigr]_{i,k}
&:=(I_L\otimes\bra{i}_S)Y(I_L\otimes\ket{k}_S),
\label{eq:sm-block-definition}\\
\left\lVert P_0^{\perp}\dket{Y}\right\rVert_2^2
&=\sum_{i,k}\left\lVert\tracelesspart{\bigl[Y\bigr]_{i,k}}\right\rVert_2^2.
\label{eq:sm-block-norm}
\end{align}
\end{subequations}
The norm identity follows from
$P_0\dket{Y}=\dket{I_L\otimes\trL(Y)/d_L}$.
The real-linear map in Lemma~\ref{lem:sm-tangent-characterization} has the
explicit form
\begin{equation}\label{eq:sm-linearized-residual}
\LMap_0(\dot{\Dtilde},\dot V)
=\left.\frac{\mathrm d}{\mathrm dt}
\ket{r(0,x(t))}\right|_{t=0}
=\sum_{m=0}^{d_E-1}\ket{m}_E\otimes
P_0^{\perp}\dket{
\dot V\bigl[\Dtilde_0\bigr]_m
+V_0\bigl[\dot{\Dtilde}\bigr]_m}.
\end{equation}
By Eq.~\eqref{eq:sm-optimized-gap-leakage-a}, its squared norm divided by
$d_L^2$ is the quadratic coefficient. Substituting
$\bigl[\dot VD_j\bigr]_{i,k}=\delta_{jk}\bigl[\dot V\bigr]_i$ and
$\bigl[V_0\bigl[\dot{\Dtilde}\bigr]_m\bigr]_{i,k}
=\delta_{i0}\bigl[\dot{\Dtilde}\bigr]_{m,k}$ into
Eq.~\eqref{eq:sm-block-norm} gives
\begin{equation}\label{eq:sm-kernel-equations}
\tracelesspart{\bigl[\dot V\bigr]_i}=0
\quad(1\leq i<d_L),
\qquad
\tracelesspart{\delta_{jk}\bigl[\dot V\bigr]_0
+\bigl[\dot{\Dtilde}\bigr]_{j,k}}=0,
\qquad
\tracelesspart{\bigl[\dot{\Dtilde}\bigr]_{m,k}}=0
\quad(d_L\leq m<d_E),
\end{equation}
where $0\leq j,k<d_L$. Together with
Eq.~\eqref{eq:sm-tangent-constraints}, the complete solution is
\begin{equation}\label{eq:sm-kernel-solutions}
\begin{aligned}
\bigl[\dot V\bigr]_i&=c_iI_L
&& (1\leq i<d_L),\\
\bigl[\dot{\Dtilde}\bigr]_{m,k}&=
\begin{cases}
-\delta_{mk}\bigl[\dot V\bigr]_0+\alpha_{mk}I_L,&0\leq m<d_L,\\
b_{mk}I_L,&d_L\leq m<d_E,
\end{cases}
&& (0\leq k<d_L),
\end{aligned}
\end{equation}
where $\bigl[\dot V\bigr]_0$ and the matrix of scalars
$(\alpha_{mk})_{m,k<d_L}$ are
anti-Hermitian, while $c_i,b_{mk}\in\Cbb$ are arbitrary.

These solutions are precisely the tangent vectors to the exact zero set. Let
$\ket{\zeta}_S:=\sum_i\zeta_i\ket{i}_S$ be a normalized-state variation at
$\ket{0}$. Normalization gives $\Re\zeta_0=0$, and the common phase redundancy
in Eq.~\eqref{eq:sm-zero-set} permits the gauge $\zeta_0=0$. Choose
$\dot U_{S\to E}:\HS\to\HE$ with
$\hat{\mu}^{\dagger}\dot U_{S\to E}+
\dot U_{S\to E}^{\dagger}\hat{\mu}=0$ and write
$\bigl[\dot U_{S\to E}\bigr]_{m,k}:=
\bra{m}_E\dot U_{S\to E}\ket{k}_S$. Differentiating the zero-set
parametrization gives
\begin{equation}\label{eq:sm-exact-zero-tangent}
\bigl[\dot V\bigr]_i=\zeta_iI_L+\delta_{i0}\bigl[\dot V\bigr]_0,
\qquad
\bigl[\dot{\Dtilde}\bigr]_{m,k}=-\delta_{mk}\bigl[\dot V\bigr]_0
+\bigl[\dot U_{S\to E}\bigr]_{m,k}I_L.
\end{equation}
Taking $\zeta_0=0$, $\zeta_i=c_i$ for $1\leq i<d_L$,
$\bigl[\dot U_{S\to E}\bigr]_{m,k}=\alpha_{mk}$ for
$0\leq m,k<d_L$, and
$\bigl[\dot U_{S\to E}\bigr]_{m,k}=b_{mk}$ for $d_L\leq m<d_E$ and
$0\leq k<d_L$ realizes
every solution in Eq.~\eqref{eq:sm-kernel-solutions}. These variations are
integrated by normalized-state, unitary-exponential, and Stiefel curves in
the three parameters of Eq.~\eqref{eq:sm-zero-set}, proving
$\ker(\LMap_0)\subseteq T_{x_0}\StiefelProd_0$. Conversely, differentiating
$\ket{r(0,x(t))}=0$ along any zero-set curve gives the reverse inclusion and
establishes Eq.~\eqref{eq:sm-kernel-tangent-space}.

Finally, Eq.~\eqref{eq:sm-linearized-residual} gives
$\ket{r(0,x(t))}=t\LMap_0(\dot x(0))+o(t)$. Together with
Eqs.~\eqref{eq:sm-optimized-gap-leakage-a} and~\eqref{eq:sm-kernel-tangent-space},
this proves the quadratic expansion and
strict positivity stated in the lemma, including for variations of the padded
inactive Kraus operators.
\end{proof}

For matrix pairs $\xi=(\xi_{\Dtilde},\xi_V)$ and
$\eta=(\eta_{\Dtilde},\eta_V)$, use the real Frobenius inner product
\begin{equation}\label{eq:sm-real-frobenius-inner-product}
\langle\xi,\eta\rangle_{\mathrm F,\Rbb}
:=\Re\tr(\xi_{\Dtilde}^{\dagger}\eta_{\Dtilde})
+\Re\tr(\xi_V^{\dagger}\eta_V),
\end{equation}
with corresponding norm $\lVert\cdot\rVert_{\mathrm F}$. For
$x\in\StiefelProd$, define its ambient distance from the zero set by
$\operatorname{dist}_{\mathrm F}(x,\StiefelProd_0)
:=\min_{y\in\StiefelProd_0}\lVert x-y\rVert_{\mathrm F}$, and at the canonical
pair define
\begin{equation}\label{eq:sm-noise-direction}
\ket{\partial_p r}
:=\left.\frac{\partial}{\partial p}\ket{r(p,x_0)}\right|_{p=0}.
\end{equation}

\begin{lemma}[localization and tangent reduction]
\label{lem:sm-localization}
If $x_p$ is any global minimizer of
$x\mapsto\lVert\ket{r(p,x)}\rVert_2^2$ on $\StiefelProd$, then it approaches
the exact zero set at the linear rate
\begin{equation}\label{eq:sm-localization-rates}
\operatorname{dist}_{\mathrm F}(x_p,\StiefelProd_0)=\OO(p).
\end{equation}
This localization yields the lower bound over the tangent space
\begin{equation}\label{eq:sm-global-liminf}
\liminf_{p\to0^+}\frac{\Deltaopt(\tau,\NN_p)}{p^2}
\geq
\frac{1}{d_L^2}
\min_{v\in\TCal_0}
\lVert\ket{\partial_p r}+\LMap_0(v)\rVert_2^2.
\end{equation}
Here $v=(\dot{\Dtilde},\dot V)\in\TCal_0$ denotes a generic
dilation--encoder tangent pair, so
$\LMap_0(v)=\LMap_0(\dot{\Dtilde},\dot V)$, as shown in
Eq.~\eqref{eq:sm-linearized-residual}.
\end{lemma}

\begin{proof}
Let $\TCal_0^{\mathrm{nor}}
:=\TCal_0\cap(T_{x_0}\StiefelProd_0)^{\perp}$. This is the
normal space of $\StiefelProd_0$ in $\StiefelProd$ at $x_0$, and the
orthogonal complement is taken with respect to the real Frobenius inner
product in Eq.~\eqref{eq:sm-real-frobenius-inner-product}.
By Lemma~\ref{lem:sm-tangent-characterization}, the restriction of $\LMap_0$
to $\TCal_0^{\mathrm{nor}}$ is injective. This space is nonzero for
$d_L\geq2$, as seen by taking $\dot{\Dtilde}=0$, a nonzero traceless
$\bigl[\dot V\bigr]_1$, and all other blocks of $\dot V$ zero. Its finite
dimensionality therefore gives a positive smallest singular value $\sigma_*$:
\begin{equation}\label{eq:sm-transverse-singular-value}
\lVert\LMap_0(v)\rVert_2\geq\sigma_*\lVert v\rVert_{\mathrm F}
\qquad(v\in\TCal_0^{\mathrm{nor}}).
\end{equation}

Compactness of $\StiefelProd$ gives a global minimizer for every $p$. Since
$\ket{r(0,x_0)}=0$, Taylor expansion gives
$\ket{r(p,x_0)}=p\ket{\partial_p r}+\OO(p^2)$. Hence every global minimizer
$x_p$ obeys $\lVert\ket{r(p,x_p)}\rVert_2=\OO(p)$ and
$\Deltaopt(\tau,\NN_p)=\OO(p^2)$. Uniform smoothness on $\StiefelProd$ then gives
$\lVert\ket{r(0,x_p)}\rVert_2=\OO(p)$. The residual at $p=0$ has a positive
minimum outside any fixed neighborhood of its compact zero set, so
$\operatorname{dist}_{\mathrm F}(x_p,\StiefelProd_0)\to0$.
Choose a closest point $\bar x_p\in\StiefelProd_0$ and a group element
$\boldsymbol g_p$ such that $\boldsymbol g_p(\bar x_p)=x_0$, and set
$x_p':=\boldsymbol g_p(x_p)=(\Dtilde_p',V_p')$. The isometry and residual
equivariance of the group action preserve global optimality and give
$\lVert x_p'-x_0\rVert_{\mathrm F}
=\lVert x_p-\bar x_p\rVert_{\mathrm F}
=\operatorname{dist}_{\mathrm F}(x_p,\StiefelProd_0)
=\operatorname{dist}_{\mathrm F}(x_p',\StiefelProd_0)$.
Thus $x_0$ is a closest zero to $x_p'$. These choices are made separately for
each $p$ and need not be smooth. Set
$\delta x_p=(\delta\Dtilde_p,\delta V_p)
:=(\Dtilde_p'-\Dtilde_0,V_p'-V_0)$, so
$\lVert\delta x_p\rVert_{\mathrm F}
=\operatorname{dist}_{\mathrm F}(x_p',\StiefelProd_0)$.
For every smooth $c(t)\in\StiefelProd_0$ through $x_0$, closest-point
optimality gives
\begin{equation*}
0=\left.\frac{\mathrm d}{\mathrm dt}
\lVert x_p'-c(t)\rVert_{\mathrm F}^2\right|_{t=0}
=-2\langle\delta x_p,\dot c(0)\rangle_{\mathrm F,\Rbb}.
\end{equation*}
Hence $\delta x_p\perp T_{x_0}\StiefelProd_0$.

Although the chord $\delta x_p$ need not lie in $\TCal_0$, it differs from a
tangent pair only at second order. For either Stiefel factor
$A\in\{\Dtilde,V\}$, write $\delta A_p:=A_p'-A_0$ and set
\begin{equation*}
S_{A,p}:=\frac{1}{2}(A_0^{\dagger}\delta A_p
+\delta A_p^{\dagger}A_0)
=-\frac{1}{2}\delta A_p^{\dagger}\delta A_p,
\qquad
v_{A,p}:=\delta A_p-A_0S_{A,p}.
\end{equation*}
The second equality follows from the Stiefel constraints at $A_p'$ and $A_0$.
Thus $v_p:=(v_{\Dtilde,p},v_{V,p})\in\TCal_0$ and
$\lVert\delta x_p-v_p\rVert_{\mathrm F}
=\OO(\lVert\delta x_p\rVert_{\mathrm F}^2)$. Decompose
$v_p=v_p^{\parallel}+v_p^{\perp}$ with
$v_p^{\parallel}\in T_{x_0}\StiefelProd_0$ and
$v_p^{\perp}\in\TCal_0^{\mathrm{nor}}$. Since
$\delta x_p\perp T_{x_0}\StiefelProd_0$, orthogonal projection gives
$\lVert v_p^{\parallel}\rVert_{\mathrm F}
=\OO(\lVert\delta x_p\rVert_{\mathrm F}^2)$, and therefore
\begin{equation}\label{eq:sm-normal-chord-estimate}
\lVert\delta x_p-v_p^{\perp}\rVert_{\mathrm F}
=\OO(\lVert\delta x_p\rVert_{\mathrm F}^2),
\qquad
v_p^{\perp}\in\TCal_0^{\mathrm{nor}}.
\end{equation}

Smooth ambient Taylor expansion, followed by
Eqs.~\eqref{eq:sm-kernel-tangent-space} and~\eqref{eq:sm-normal-chord-estimate},
gives
\begin{equation*}
\ket{r(0,x_p')}
=\LMap_0(v_p^{\perp})+\OO(\lVert\delta x_p\rVert_{\mathrm F}^2).
\end{equation*}
Together with Eq.~\eqref{eq:sm-transverse-singular-value} and
$\lVert\ket{r(0,x_p')}\rVert_2=\OO(p)$, this gives, for a constant $C$
independent of $p$,
\begin{equation*}
\lVert\delta x_p\rVert_{\mathrm F}
\leq\lVert v_p^{\perp}\rVert_{\mathrm F}
+C\lVert\delta x_p\rVert_{\mathrm F}^2,
\qquad
\sigma_*\lVert v_p^{\perp}\rVert_{\mathrm F}
\leq Cp+C\lVert\delta x_p\rVert_{\mathrm F}^2.
\end{equation*}
Since $\lVert\delta x_p\rVert_{\mathrm F}\to0$, the quadratic terms can be
absorbed. Group invariance then proves Eq.~\eqref{eq:sm-localization-rates}
for every global minimizer, without choosing a smooth optimizer branch.

Choose a sequence $p\to0^+$ that realizes the limit inferior. From
Eq.~\eqref{eq:sm-localization-rates} and the normal chord estimate in
Eq.~\eqref{eq:sm-normal-chord-estimate},
$\lVert v_p^{\perp}\rVert_{\mathrm F}=\OO(p)$. Thus, along a subsequence,
$v_p^{\perp}/p\to v^{\perp}\in\TCal_0^{\mathrm{nor}}$, and joint Taylor
expansion gives
\begin{equation}\label{eq:sm-scaled-residual-expansion}
\frac{1}{p}\ket{r(p,x_p')}
=\ket{\partial_p r}
+\LMap_0\left(\frac{v_p^{\perp}}{p}\right)+\OO(p).
\end{equation}
Taking norms in Eq.~\eqref{eq:sm-scaled-residual-expansion} and using
Eq.~\eqref{eq:sm-optimized-gap-leakage-b} bounds the limit inferior below by
$d_L^{-2}$ times the minimum over $\TCal_0^{\mathrm{nor}}$. Because
$\TCal_0=(T_{x_0}\StiefelProd_0)\oplus\TCal_0^{\mathrm{nor}}$ and $\LMap_0$
vanishes on $T_{x_0}\StiefelProd_0$, this minimum agrees with the one over
$\TCal_0$, proving Eq.~\eqref{eq:sm-global-liminf}. The normal-space minimum
is attained because Eq.~\eqref{eq:sm-transverse-singular-value} makes its
objective coercive.
\end{proof}

\medskip
\noindent\emph{Evaluation and saturation.}
Lemma~\ref{lem:sm-localization} reduces the global problem to the last
minimum in
Eq.~\eqref{eq:sm-global-liminf}. Write
$G_{ij}:=I_L\otimes\ket{i}_S\!\bra{j}_S=\sqrt{d_L}B_{ij}(0)$ and
\begin{equation}\label{eq:sm-generator-expansion}
\beta:=\frac{d_L^2}{2(d_L^2-1)},
\qquad
\hat{\pi}_C^{\circ}:=\hat{\pi}_C-\frac{I_C}{d_L^2},
\qquad
\sqrt{d_L}Q_p=I_C-p\beta\hat{\pi}_C^{\circ}+\OO(p^2).
\end{equation}
Using $B_{ij}(p)=C_i(p)D_j$ and Eq.~\eqref{eq:sm-family-pair}, we obtain
$B_{ij}(p)=Q_pG_{ij}$, so
Eq.~\eqref{eq:sm-generator-expansion} gives
$\sqrt{d_L}B_{ij}(p)
=(I_C-p\beta\hat{\pi}_C^{\circ})G_{ij}+\OO(p^2)$.
Differentiating the frame relation
$P_p\dket{B_{0j}(p)}=\dket{B_{0j}(p)}$ gives
$\left.\partial_pP_p^{\perp}\right|_{p=0}\dket{G_{0j}}
=\beta P_0^{\perp}\dket{\hat{\pi}_C^{\circ}G_{0j}}$.
When differentiating the residual at $x_0$, the product rule term containing
$\left.\partial_pM_p^{1/2}\right|_{p=0}$ acts on
$P_0^{\perp}\dket{V_0\bigl[\Dtilde_0\bigr]_m}=0$ and therefore vanishes
for every $m$.
Consequently,
\begin{equation}\label{eq:sm-residual-derivative}
\ket{\partial_p r}
=\beta\sum_{j=0}^{d_L-1}\ket{j}_E\otimes
P_0^{\perp}\dket{\hat{\pi}_C^{\circ}G_{0j}}.
\end{equation}

Set $\pi_{0,L}^{\circ}:=\kb{0}_L-I_L/d_L$. The block definition in
Eq.~\eqref{eq:sm-block-definition} gives
\begin{equation}\label{eq:sm-generator-blocks}
\bigl[\hat{\pi}_C^{\circ}G_{0j}\bigr]_{i,k}
=\delta_{jk}\left[
\frac{1}{d_L}\ket{i}_L\!\bra{0}_L
-\frac{\delta_{i0}}{d_L^2}I_L\right].
\end{equation}
Combining Eq.~\eqref{eq:sm-block-norm} with the definition of $\LMap_0$ and
Eq.~\eqref{eq:sm-residual-derivative} gives
\begin{equation}\label{eq:sm-tangent-residual}
\begin{aligned}
\lVert\ket{\partial_p r}
+\LMap_0(\dot{\Dtilde},\dot V)\rVert_2^2
&=d_L\sum_{i=1}^{d_L-1}
\left\lVert\tracelesspart{\bigl[\dot V\bigr]_i}
+\frac{\beta}{d_L}\ket{i}_L\!\bra{0}_L\right\rVert_2^2\\
&\quad+\sum_{j,k=0}^{d_L-1}
\left\lVert\tracelesspart{\delta_{jk}\bigl[\dot V\bigr]_0
+\bigl[\dot{\Dtilde}\bigr]_{j,k}}
+\delta_{jk}\frac{\beta}{d_L}\pi_{0,L}^{\circ}\right\rVert_2^2
+\sum_{m=d_L}^{d_E-1}\sum_{k=0}^{d_L-1}
\left\lVert\tracelesspart{\bigl[\dot{\Dtilde}\bigr]_{m,k}}\right\rVert_2^2.
\end{aligned}
\end{equation}
The factor $d_L$ in the first sum counts the active environment rows, while the
Kronecker deltas restrict the forced defect in the second sum to $j=k$.
The first and third sums and the off-diagonal terms of the second sum are
nonnegative and admit zero. For a diagonal term, set
$A_j:=\tracelesspart{[\dot V]_0+[\dot{\Dtilde}]_{j,j}}$. The tangent
constraints make $A_j$ anti-Hermitian, whereas $\pi_{0,L}^{\circ}$ is
Hermitian and traceless, so
$\Re\tr[A_j^{\dagger}\pi_{0,L}^{\circ}]=0$. Thus each diagonal term is at
least $\beta^2\lVert\pi_{0,L}^{\circ}\rVert_2^2/d_L^2$. The $d_L$ diagonal
blocks and $\lVert\pi_{0,L}^{\circ}\rVert_2^2=(d_L-1)/d_L$ give the lower
bound $\beta^2(d_L-1)/d_L^2$.

The bounds are simultaneously attained by the feasible tangent
\begin{equation}\label{eq:sm-minimizing-tangent}
\dot{\Dtilde}=0,\qquad
\bigl[\dot V\bigr]_0=0,\qquad
\bigl[\dot V\bigr]_i=-\frac{\beta}{d_L}\ket{i}_L\!\bra{0}_L
\quad(1\leq i<d_L).
\end{equation}
It satisfies Eq.~\eqref{eq:sm-tangent-constraints}, makes the first and third
sums and all off-diagonal terms zero, and has $A_j=0$ for every $j$. Hence
\begin{equation}\label{eq:sm-tangent-minimum}
\min_{(\dot{\Dtilde},\dot V)\in\TCal_0}
\lVert\ket{\partial_p r}
+\LMap_0(\dot{\Dtilde},\dot V)\rVert_2^2
=\frac{\beta^2(d_L-1)}{d_L^2}
=\frac{d_L^2(d_L-1)}{4(d_L^2-1)^2}.
\end{equation}
This positive coefficient is the squared norm of the Hermitian component that
no feasible anti-Hermitian first-order correction can cancel.

This minimizing tangent is realized by an explicit feasible path. For a
normalized state $\ket{\chi}_S\in\HS$, set
\begin{equation}\label{eq:sm-polar-initializer}
C_p(\chi):=Q_p(I_L\otimes\ket{\chi}_S),
\qquad
V_p(\chi):=C_p(\chi)
[C_p(\chi)^{\dagger}C_p(\chi)]^{-1/2}.
\end{equation}
By unitary freedom in Kraus representations, $C_p(\chi)$ is a Kraus operator
in an equivalent Kraus representation of $\CCal_p^{\opt}$. Since $Q_p\succ0$,
$C_p(\chi)$ has full column rank and
$V_p(\chi)^{\dagger}V_p(\chi)=I_L$. Hence $\CCal_{V_p(\chi)}$ is a feasible
partial-isometry encoder. At $\ket{\chi}_S=\ket{0}_S$,
\begin{equation}\label{eq:sm-initializer-expansion}
V_p(\chi)
=V_0-p\beta(I_C-V_0V_0^{\dagger})\hat{\pi}_C^{\circ}V_0+\OO(p^2),
\end{equation}
whose derivative is the minimizing $\dot V$ in
Eq.~\eqref{eq:sm-minimizing-tangent}. Keeping the Stinespring dilation
$\Dtilde_0$ fixed gives $\dot{\Dtilde}=0$. This trial path gives the matching
upper bound
\begin{equation}\label{eq:sm-residual-minimum-limsup}
\limsup_{p\to0^+}\frac{\Deltaopt(\tau,\NN_p)}{p^2}
\leq\frac{d_L-1}{4(d_L^2-1)^2}.
\end{equation}
Covariance transports this construction to the corresponding reference pair
for every normalized state $\ket{\chi}_S$ without changing its asymptotic
residual.

\begin{corollary}[exact quadratic coefficient]
\label{cor:sm-exact-quadratic-coefficient}
As $p\to0^+$, the optimized rank-one gap and the exact infidelity in
Eq.~\eqref{eq:sm-family-optimum} give
\begin{equation}\label{eq:sm-delta-asymptotics}
\Deltaopt(\tau,\NN_p)
=\frac{d_L-1}{4(d_L^2-1)^2}p^2+o(p^2)
\Longrightarrow
\kgap
=\frac{d_L-1}{4}\left(\frac{d_L^2-2}{d_L^2-1}\right)^2.
\end{equation}
\end{corollary}

\begin{proof}
Substituting Eq.~\eqref{eq:sm-tangent-minimum} into the global lower bound
in Eq.~\eqref{eq:sm-global-liminf} and combining it with the feasible-path
upper bound in Eq.~\eqref{eq:sm-residual-minimum-limsup} give matching $\liminf$
and
$\limsup$, proving Eq.~\eqref{eq:sm-delta-asymptotics}. Dividing by the exact
squared infidelity in Eq.~\eqref{eq:sm-family-optimum} gives
Eq.~\eqref{eq:example-kappa}.
\end{proof}

\medskip
\noindent\emph{Relation to Theorem~\ref{thm:quadratic}.}
The explicit encoder above agrees exactly with that rank-one encoder
construction. Following the notation in the proof of
Theorem~\ref{thm:quadratic}, the Kraus operator $C_p(\chi)$ from
Eq.~\eqref{eq:sm-polar-initializer} has input weight $w_{\chi}=1/d_L$ and
relative fidelity defect
$x_{\chi}=1-\lambda_1=\epsopt(\tau,\NN_p)\leq1/d_L$. Thus the first polar
decomposition in that construction yields $V_p(\chi)$. Let
$\ket{\chi^*}_L$ denote the entrywise complex conjugate of $\ket{\chi}_S$ in
the bases identified by $\ket{i}_S\mapsto\ket{i}_L$, and set
$\Pi_{\chi^*}:=\kb{\chi^*}_L$. Let
$\Ktilde_{p,\chi}$ denote the contraction in the second step and write its
polar decomposition as
$\Ktilde_{p,\chi}=H_{p,\chi}U_{p,\chi}^{\dagger}$, where
$U_{p,\chi}$ is unitary. Invariance under unitary
changes of the composite Kraus representation and covariance give
\begin{equation}\label{eq:sm-family-second-polar}
\begin{aligned}
\Ktilde_{p,\chi}
&=k_{\parallel}(p)\Pi_{\chi^*}
+k_{\perp}(p)(I_L-\Pi_{\chi^*}),&
k_{\parallel}(p)^2
&=\frac{d_L\lambda_1(1+d_L\lambda_5)}{d_L+1},&
k_{\perp}(p)^2
&=\frac{d_L^2\lambda_1\lambda_4}{d_L^2-1}.
\end{aligned}
\end{equation}
The real coefficients $k_{\parallel}(p)$ and $k_{\perp}(p)$ are continuous
and equal to one at $p=0$. Their squared values remain positive for
$0\leq p<1$ because
$\lambda_1,\lambda_4,\lambda_5>0$. Hence they are positive throughout this
interval, so $\Ktilde_{p,\chi}\succ0$ and $U_{p,\chi}=I_L$. Therefore,
$W_p(\chi):=V_p(\chi)U_{p,\chi}=V_p(\chi)$. Thus
$(\CCal_{V_p(\chi)},\trS)$ realizes the globally
optimal quadratic coefficient of the rank-one gap. Higher order corrections
to the rank-one optimum remain undetermined, so no closed form for
$\Fopto(\tau,\NN_p)$ or $\Deltaopt(\tau,\NN_p)$ is obtained. Exact optimality
of this pair at fixed $p>0$ remains open.

\section{Numerical optimization and reproducibility}\label{sec:sm-algorithms}

In terms of the dressed Choi matrix in Eq.~\eqref{eq:sm-dressed-choi}, the
objective has two equivalent forms
\begin{equation}\label{eq:sm-dressed-objective}
\Fe(\rho,\DD\circ\NN\circ\CCal)
=\tr[J(\DD^{\dagger})J_{\rho^2}(\NN\circ\CCal)]
=\tr[J_{\rho^2}((\DD\circ\NN)^{\dagger})J(\CCal)].
\end{equation}
Both algorithms accept arbitrary feasible initial pairs. In multistart runs,
we retain the best feasible pair.

\subsection{Rank Unrestricted Optimization}\label{sec:sm-unrestricted-algorithm}

Completing the marginal of a CPTNI Choi matrix to equality adds a positive
semidefinite contribution to either objective in
Eq.~\eqref{eq:sm-dressed-objective}. Hence
each frozen optimum has a CPTP representative, and
Algorithm~\ref{alg:sm-unrestricted} alternates semidefinite programs~(SDPs)
with equality constraints.

\begin{algorithm}[H]
\caption{Unrestricted seesaw algorithm for $\Fopt(\rho,\NN)$.}
\label{alg:sm-unrestricted}
\small
\begin{algorithmic}[1]
\Require state $\rho$, CPTNI noise map $\NN$, starting feasible CPTP encoder
$J(\CCal)$ and decoder $J(\DD^{\dagger})$, tolerance
$\eps_{\mathrm{tol}}\in(0,1)$, and $k_{\max}\in\mathbb{N}^{+}$
\Ensure encoder $J(\CCal)$, decoder $J(\DD^{\dagger})$, and fidelity $F$
\State $k\gets0$ and $F_{\mathrm{old}}\gets-\infty$
\State $F\gets\Fe(\rho,\DD\circ\NN\circ\CCal)$
\While{$k<k_{\max}$ and $F-F_{\mathrm{old}}\geq\eps_{\mathrm{tol}}$ and $F<1$}
  \State $k\gets k+1$ and $F_{\mathrm{old}}\gets F$
  \State $K_{\DD}\gets J_{\rho^2}(\NN\circ\CCal)$
  \State $J(\DD^{\dagger})\gets
  \operatorname*{argmax}\{\tr(YK_{\DD}):Y\succeq0,\trL Y=I_C\}$
  \State $K_{\CCal}\gets J_{\rho^2}((\DD\circ\NN)^{\dagger})$
  \State $J(\CCal)\gets
  \operatorname*{argmax}\{\tr(YK_{\CCal}):Y\succeq0,\trC Y=I_L\}$
  \State $F\gets\Fe(\rho,\DD\circ\NN\circ\CCal)$
\EndWhile
\State \Return $J(\CCal)$, $J(\DD^{\dagger})$, and $F$
\end{algorithmic}
\end{algorithm}

Each SDP update is globally optimal, so $F$ is nondecreasing. The joint
problem is nonconvex, and convergence gives no global certificate. If $\rho$
is rank deficient, the objective ignores $\ker\rho$, although the
quantum-channel constraints still act on all of $\HL$.

\subsection{Optimization with a rank-one encoder}
\label{sec:sm-rank-one-algorithm}

For the rank-one restriction, we use the partial-isometry parametrization and
reduction in Lemma~\ref{lem:sm-partial-isometry-reduction}, together with the
frozen objective in Eq.~\eqref{eq:sm-rank-one-objective}.

\begin{algorithm}[H]
\caption{Rank-one restricted seesaw algorithm for
$\Fopto(\rho,\NN)$.}\label{alg:sm-rank-one}
\small
\begin{algorithmic}[1]
\Require state $\rho$, CPTNI noise map $\NN$, starting partial isometry
$V=V\PiRho$ with $V^{\dagger}V\preceq\PiRho$, starting feasible CPTP decoder
$J(\DD^{\dagger})$, tolerance $\eps_{\mathrm{tol}}\in(0,1)$, and
$k_{\max}\in\mathbb{N}^{+}$
\Ensure support partial isometry $V$, decoder $J(\DD^{\dagger})$, and fidelity
$F$
\State $k\gets0$ and $F_{\mathrm{old}}\gets-\infty$
\State $F\gets\Fe(\rho,\DD\circ\NN\circ\CCal_V)$
\While{$k<k_{\max}$ and $F-F_{\mathrm{old}}\geq\eps_{\mathrm{tol}}$ and $F<1$}
  \State $k\gets k+1$ and $F_{\mathrm{old}}\gets F$
  \State $K_{\DD}\gets J_{\rho^2}(\NN\circ\CCal_V)$
  \State $J(\DD^{\dagger})\gets
  \operatorname*{argmax}\{\tr(YK_{\DD}):Y\succeq0,\trL Y=I_C\}$
  \State $K_{\CCal}\gets J_{\rho^2}((\DD\circ\NN)^{\dagger})$
  \State $\widetilde{V}\gets\operatorname{unvec}(K_{\CCal}\dket{V})\PiRho$
  \State $V\gets\widetilde{V}
  (\widetilde{V}^{\dagger}\widetilde{V})^{-1/2}$
  \Comment{Moore--Penrose inverse square root}
  \State $F\gets\dbra{V}K_{\CCal}\dket{V}$
\EndWhile
\State \Return $V$, $J(\DD^{\dagger})$, and $F$
\end{algorithmic}
\end{algorithm}

Here ``$\operatorname{unvec}$'' inverts $\dket{\cdot}$. Lines 8 and 9 implement the
support-restricted Reimpell--Werner iteration~(RW) for a single Kraus
operator~\cite{Reimpell2005}. The RW iteration is known to resemble a power
method and does not increase the Choi matrix rank. For a compact singular value
decomposition
$\widetilde{V}=U\Sigma W^{\dagger}$, the Moore--Penrose inverse square root is
$(\widetilde{V}^{\dagger}\widetilde{V})^{-1/2}
=W\Sigma^{-1}W^{\dagger}$. Thus the update is
$V_{\mathrm{new}}:=\widetilde{V}(\widetilde{V}^{\dagger}\widetilde{V})^{-1/2}
=UW^{\dagger}$. It is the polar partial isometry and satisfies
$V_{\mathrm{new}}^{\dagger}V_{\mathrm{new}}\preceq\PiRho$. It also maximizes
$\Re\tr[(V')^{\dagger}\widetilde{V}]$ over
$(V')^{\dagger}V'\preceq\PiRho$. Because $V$ is feasible for this
maximization, setting $\delta V:=V_{\mathrm{new}}-V=\delta V\PiRho$ gives
$\Re\dbra{\delta V}K_{\CCal}\dket{V}
=\Re\tr(\delta V^{\dagger}\widetilde{V})\geq0$.
Expanding the objective at $V_{\mathrm{new}}$ gives
\begin{equation}
\widetilde{F}(V_{\mathrm{new}})-\widetilde{F}(V)
=2\Re\dbra{\delta V}K_{\CCal}\dket{V}
+\dbra{\delta V}K_{\CCal}\dket{\delta V}\geq0.
\end{equation}
Here the first term is nonnegative by the polar maximization above, and the
second is nonnegative because $K_{\CCal}\succeq0$. The decoder SDP half-step
also cannot decrease the fidelity. Thus the fidelity values, though not
necessarily the encoder--decoder iterates, form a nondecreasing sequence
bounded by one and hence converge.

\subsection{Numerical implementation}\label{sec:sm-numerics}

For the family $(\tau,\NN_p)$, Fig.~\ref{fig:gap}(b) uses logarithmic sweeps
toward $p=0$ and linear sweeps toward $p=1$ for $d_L=2,\ldots,11$. The
unrestricted optimum $\Fopt$ is known analytically from
Eq.~\eqref{eq:sm-family-optimum}. The rank-one curves are obtained by
multistart runs of the rank-one restricted seesaw algorithm
(Algorithm~\ref{alg:sm-rank-one}), using the RW iteration for the decoder
half-step and a singular value decomposition for the encoder update.
Interior point algorithms independently cross-check selected decoder
subproblems. For each run, the initial encoder is sampled at random from the
Stiefel manifold of partial isometries.

For any rank-one partial-isometry encoder $\CCal_1$ used in the sweep, let
$\DD$ be its reoptimized TP decoder and set $\BCal:=\CCal_1\circ\DD$.
Because $\BCal$ is TP, Eq.~\eqref{eq:example-certificate} gives the plotted
loss as $\Delta_{\mathrm{plot}}=\tr[Z_pJ(\BCal)]/d_L^2$. This trace expression avoids
subtracting fidelities close to one, and the curves stop before
$\Delta_{\mathrm{plot}}$ becomes comparable to double precision roundoff.
The channel fidelity attained by every feasible numerical pair is a lower
bound on $\Fopto$, so the corresponding $\Delta_{\mathrm{plot}}$ is an upper
bound on $\Deltaopt$. Runs from multiple initializations converge to the same
value of $\Delta_{\mathrm{plot}}$ within a numerical uncertainty of $\pm10^{-12}$. This
agreement supports the numerical stability of the plotted estimates but does
not provide an analytic global certificate at fixed $p$.
Independently, the rank obstruction and
Eq.~\eqref{eq:sm-delta-asymptotics} establish the positive gap and its
asymptotic coefficient.

\end{document}